\documentclass[11pt]{article}
\usepackage{amsfonts}
\usepackage{natbib}
\usepackage{amsmath,amssymb,array,graphicx}
\usepackage{enumerate}
\usepackage{graphicx}
\usepackage{multirow}
\usepackage{subfigure}
\usepackage{leftidx}
\usepackage{setspace}
\usepackage{tikz}
\usepackage{amsthm}
\usepackage[colorlinks,linkcolor=black,anchorcolor=black,citecolor=black]{hyperref}
\usetikzlibrary{shapes.geometric, arrows}
\usepackage[titletoc,toc,title]{appendix}
\usepackage{xcolor}

\theoremstyle{definition}
\newtheorem{corollary}{Corollary}

\newtheorem{definition}{Definition}
\newtheorem{theorem}{Theorem}
\newtheorem{lemma}{Lemma}
\newtheorem{remark}{Remark}
\newtheorem{assumption}{Assumption}
\newtheorem{proposition}{Proposition}

\newcommand{\beq}{\begin{equation}}
\newcommand{\eeq}{\end{equation}}
\newcommand{\bey}{\begin{eqnarray}}
\newcommand{\eey}{\end{eqnarray}}
\newcommand{\nn}{\nonumber}

\def\ds{\displaystyle}

\def\i{{\infty}}

\def\t{\tau}

\usepackage{setspace}
\numberwithin{equation}{section}

\begin{document}
\newpage
\title{Valuation of contingent claims with short selling bans under an equal-risk pricing framework}
\author{
 Guiyuan Ma\footnote{Corresponding author.
 School of Economics and Finance, Xi'an Jilting University, Xi'an,  China. Email:guiyuanma@xjtu.edu.cn.}
 \and
 Song-Ping Zhu\footnote{
School of Mathematics and Applied Statistics,
 University of Wollongong, NSW 2522, Australia.
Email: spz@uow.edu.au.}
\and
 Ivan Guo\footnote{
School of Mathematics,
 Monash University, VIC 3800, Australia. Monash Centre for Quantitative Finance and Investment Strategies.
 Email:ivan.guo@monash.edu.
}
}
\date{}
\maketitle
\begin{abstract}
This paper studies the valuation of European contingent claims with short selling bans under the equal risk pricing (ERP) framework proposed in \cite{guo2017equal} where analytical pricing formulae were derived in the case of monotonic payoffs under risk-neutral measures.
  We establish a unified framework for this  new pricing approach so that its range of application can be significantly expanded. The results of \cite{guo2017equal} are extended to the case of non-monotonic payoffs (such as a butterfly spread option) under risk-neutral measures. We also provide numerical schemes for computing equal-risk prices under other measures such as the original physical measure. Furthermore, we demonstrate how  short selling bans can affect the valuation of contingent claims  by comparing  equal-risk prices with Black-Scholes prices.
\end{abstract}
\smallskip\noindent
{\bf Keywords.} Equal-risk pricing (ERP); Short selling bans; Hamilton-Jacobi-Bellman (HJB) equation; Non-monotonic payoff.

\section{Introduction}
During the Global Financial Crisis in 2007--2009, most regulatory authorities around the world imposed restrictions or bans on short selling to reduce the volatility of financial markets and to limit the negative impacts of downturn markets \citep{beber2013short}. These interventions were implemented with an intention to  prevent further drops of stock prices. However, these regulations imposed on short selling also resulted in some new problems, one of which is the valuation of contingent claims in a market where short selling is partially restricted or completely banned.  In the literature, the effects of short selling restrictions on stock prices and the valuation of contingent claims have been studied extensively \citep{figlewski1981informational,jones2002short,avellaneda2009dynamic,ma2018pricing,chen2018numerical,ma2019pricing,atmaz2019option,ma2020convergence,beber2021short,shi2021dynamic}. In this paper, we focus on the valuation of contingent claims in a financial market where short selling is completely banned.

According to the fundamental theorem of asset pricing \citep{shreve2004stochastic}, every contingent claim can be replicated perfectly by some self-financing hedging strategy in a complete market and the price of the contingent claim must equal to the cost of constructing such a portfolio according to no-arbitrage arguments. However,  in incomplete markets where short selling is absent, such perfect hedging strategies are not always available. In the literature, the valuation of contingent claims in an incomplete market has also been explored extensively and a large number of approaches and techniques have been proposed. Generally, these related literature can be grouped into two categories.

Papers in the first category share a common feature that an equivalent martingale measure is chosen as the pricing measure according to some optimal criteria. Since the equivalent martingale measure is not unique in the incomplete market, the choice of the pricing measure varies among different studies. \cite{follmer1991hedging} first  provided a criterion to choose  a \textit{minimal martingale measure}. Then a \textit{minimal entropy martingale measure} was proposed  by \cite{frittelli2000minimal} to minimize the entropy  between the objective probability measure and the chosen risk-neutral measure.  Similar concepts, such as the \textit{minimal distance martingale measure} and the \textit{minimax measure} were also put forward by \cite{goll2001minimax} and \cite{bellini2002existence}, respectively. Each chosen pricing measure leads to a different price, which is ``fair'' according to the criteria behind the choice.  It is difficult to justify which choice of these equivalent martingale measures is most ``correct''. % Calibration has to be implemented to demonstrate that a particular choice of equivalent martingale measures is the best choice.

Papers in the second category include \cite{karatzas1996pricing}, \cite{davis1997option}, \cite{rouge2000pricing}, \cite{musiela2004example}  and \cite{hugonnier2005utility}. The key idea of these papers is the so-called \textit{utility indifference pricing}, which is characterized by an investor who chooses a utility function  to describe his risk preference. Then two concepts of ``fair price'' are introduced. The utility indifference buying price $p^b$ is the price at which the utility of the investor  is indifferent between (i) paying nothing and not having the claim; and (ii) paying $p^b$ now to receive the contingent claim at expiry \citep{henderson2009utility}. The utility indifference selling price is defined similarly. In the finance literature, the utility indifference price is also called the private or subjective price because such a price is derived based on the investor's own utility preference \citep{detemple1999nontraded,teplaa2000optimal}. Generally,  the utility indifference price is nonlinear  due to the concavity of the utility function, which is significantly different from the Black-Scholes price \citep{black1973pricing}.  These two prices only coincide in a complete market \citep{fleming2006controlled}.

For any contingent claim in an incomplete market, \cite{el1995dynamic} demonstrated that there exists a price interval  which avoids arbitrage opportunities. The maximum price of this interval, called the \textit{selling price}, is the lowest price that allows the seller to  superhedge the contingent claim. Similarly, the minimum price of this interval, called the \textit{buying price}, is the highest price that the buyer is willing to pay for a contingent claim while superhedging. Both of these concepts have been addressed in the literature on hedging and pricing under transaction costs \citep{hodges1989optimal,davis1997option,constantinides1999bounds,munk1999valuation}.
Both the selling price and the buying price are private prices for the respective parties as they serve to minimize unilateral risks.  In an over the counter transaction, the buyer and seller have to negotiate and compromise with each  other in order to reach a deal.

Recently, \cite{guo2017equal} proposed a completely new approach, referred to as the \textit{equal-risk pricing approach}, which determines the valuation of contingent claims by simultaneously analyzing the risk exposures for  both parties involved in the contract.  As pointed out in Remark 3.2 of \cite{guo2017equal}, it is different to the existing utility indifference pricing method.
The \textit{equal-risk price} aims to distribute the expected loss equally between the two parties. Such a price is interpreted as a fair price that both parties are willing to accept during the negotiation if they intend to enter into a derivative contract. The equal-risk price is a transactional price and it must lie in the price interval between the selling price and the buying price. Both the seller and the buyer  face the same amount of risk if they accept such a price. The existence and uniqueness of the equal-risk price has been established by \cite{guo2017equal} and they also demonstrated its consistency with the Black-Scholes price if the market is complete. Furthermore, two analytical formulae have also been derived in some special cases. However, the derivation heavily depends on the monotonicity of the payoff, which has limited its application to general contingent claims. The hedging arguments used in the equal-risk pricing approach also has some similarity to super-replication under gamma or under delta constraints \citep{Bouchard2016Hedging}.
Very recently, \cite{he2020revised} studied the case in which a ban of short selling of the underlying alone is somewhat less ``effective'' than the extreme case discussed by \cite{guo2017equal}. \cite{marzban2020equal} extended the equal risk approach to convex risk measures. \cite{carbonneau2021deep,carbonneau2021measure,carbonneau2020equal}  further develop a deep reinforcement learning approach to extend the work of \cite{guo2017equal}.

The main contribution of this paper that we establish a unified PDE framework for  the  equal-risk pricing approach so that its range of application can be significantly expanded. First, we derive the pricing formulae for European call and put options under our PDE framework, which demonstrates that it is consistent with the previous results of \cite{guo2017equal}. Next, we extend the analytical results of \cite{guo2017equal} to the case of general, non-monotonic payoff functions under the equivalent martingale measure $\mathbb{Q}$. Then, we demonstrate how to compute equal risk prices in other measures using numerical PDE approaches by pricing a butterfly spread option under the original physical measure $\mathbb{P}$. Finally, by comparing the equal-risk prices with  the Black-Scholes prices, we  numerically demonstrate  how the short selling ban affects the valuation of the butterfly spread option.

The paper is organized as follows. In Section 2, a financial market with a short selling ban is introduced  and then a PDE framework is established to derive equal-risk prices for general contingent claims. In Section 3, we derive the pricing formulae in the case of the equivalent martingale measure, for European call and put options as well as general options with non-monotonic payoffs. In Section 4, an ADI (alternative direction implicit) numerical scheme is provided to solve the PDE system and two numerical examples are demonstrated for the original physical measure. Conclusions are provided in the last section.
\section{The equal-risk pricing approach}\label{PDE_frame}
\subsection{A market model without short selling}
Consider a financial model on a  probability space $(\Omega,\mathcal{F},\mathbb{P})$. Let $\mathbb{F}=\{\mathcal{F}_t:t>0\}$ be the filtration which satisfies the usual conditions and it represents the information flow available to market participants. Let $\mathbb{P}$ denote the original physical measure in the market. We assume there are only two assets traded continuously in the market. One is a risk-free asset whose price satisfies
\begin{equation}\label{bond}
  dP_t=rP_tdt,
\end{equation}
where $r$ is the risk-free interest rate. The other one is a risky asset whose price follows
\begin{equation}\label{BS}
  dS_t=\mu S_tdt+\sigma S_tdW_t,
\end{equation}
where $\sigma$ is the volatility of the underlying and $W_t$ is a standard Brownian motion.
\begin{remark}
It has to be pointed that the choice of measure plays an important role in the valuation of contingent claims.  If there are no restrictions on  short selling,  the market is complete and there exists a unique  equivalent martingale measure $\mathbb{Q}$. Under this unique equivalent martingale measure, the parameter $\mu$ in the dynamics (\ref{BS}) should be replaced by $r$. The price of a  European contingent claim that expires at time $T$ with the payoff  $Z(S)\in \mathbb{L}^2(\Omega,\mathcal{F},\mathbb{Q})$ can be easily calculated as $v=\mathbf{E}_{\mathbb{Q}}\left[e^{-rT}Z(S_T)\right]$. This price is accepted by both the seller and the buyer since both are able to perfectly replicate the claim  using self-financing trading strategies.
\end{remark}
When the short selling  is completed banned,  the market becomes incomplete as perfect replication strategies no longer exist for some contingent claims. In this case, an admissible trading strategy is a progressively measurable non-negative process $\phi_t$, which represents the number of stock held at time $t$. Given an initial wealth $v$, consider the following self-financing trading strategy: hold $\phi_t$ shares of stock at time $t$ and keep  the remaining wealth in  the risk-free asset. Then the wealth process, denoted by $v_t$, follows
\begin{equation}\label{hedging_account1}
  d v_t=\underbrace{d(\phi_tS_t)}_{\text{stock account}}+\underbrace{d(v_t-\phi_tS_t)}_{\text{risk-free  account}}=\phi_tdS_t+r(v_t-\phi_tS_t)dt=[rv_t+(\mu-r)\phi S_t]dt+\phi_t\sigma S_tdW_t,
\end{equation}
where $\phi_t$  belongs to the set of all progressively measurable, non-negative and square integrable  trading strategies
\begin{equation}\label{admissable set}
\mathbf{\Phi}:=\left\{\phi(t,\omega):[0,T]\times \Omega\rightarrow R^+ \middle|\quad  \mathbf{E}_{\mathbb{P}}\left[\int_0^T\phi^2(t,\omega)dt\right]<\i\right\}.
\end{equation}

\subsection{The equal-risk price of general contingent claims}
Under the short selling ban, the market is incomplete. In this case, any unilateral utility-based pricing approach results in a price interval. Any price that lies strictly within this interval leads to a scenario in which both the buyer and the seller face some level of risks. Intuitively, a higher price implies a high risk exposure for the buyer; while a lower price implies a higher risk exposure for the seller.  \cite{guo2017equal} proposed a criterion to determine an equal-risk price which distributes the risk between the buyer and the seller equally. To measure the risk exposures, they introduced the following risk function.
\begin{definition}\label{def1}
A function $R:\mathbb{R}\rightarrow \mathbb{R}$ is called a \textit{risk function} if it satisfies the following conditions:
\begin{enumerate}
  \item $R(x)$ is non-decreasing, convex and has a finite lower bound $L^B$.
  \item $R(0)=0$ and $R(x)>0$ for all $x>0$.
\end{enumerate}
\end{definition}
\begin{remark}
It is obvious that both $R_1(x)=x^+$ and $R_2(x)=e^x-1$ are risk functions. The former is adopted by \cite{guo2017equal}, while the latter is the one we choose in this paper.    The smoothness of $R_2(x)$  facilitates the  derivation of pricing formula.
\end{remark}

Suppose an investor  sells one  European contingent claim with the payoff  $Z(S_T)$ at the price $v$.  After receiving the payment, the seller  establishes an hedging account with the initial wealth $v$ to hedge his future liability  $Z(S_T)$. The terminal wealth of the hedging account $v_T$ is used to reduce the risk. As a result, the minimum  risk exposure for the seller at expiry is defined by
\begin{equation}\label{seller_risk}
 \rho^s(S,v;Z)= \inf\limits_{\phi(\cdot)\in \mathbf{\Phi}}  \mathbf{E}_{\mathbb{P}}^{S,v}\left[R\left(Z(S_T)-v_T^{v,\phi(\cdot)}\right)\right],
\end{equation}
where $\mathbb{E}_{\mathbb{P}}^{S,v}$ denotes the conditional expectation under the original physical measure $\mathbb{P}$ with $S_0=S,v_0=v$ and $v_t^{v,\phi(\cdot)}$ is the solution of Equation (\ref{hedging_account1}) given the trading strategy $\phi(\cdot)$  and the initial wealth $v$.

To calculate the minimum risk exposure for the seller, it suffices to solve an optimal stochastic control problem with the objective function (\ref{seller_risk}) and  the dynamics of $S_t$ and $v_t$. By the dynamic programming principle, the value function $F^s(t,S,v)$ satisfies
\beq
\left\{
\begin{array}{l}
\ds  0=\frac{\partial F^s}{\partial t}+\inf\limits_{\phi\geq 0}\left\{\mathcal{L}_1^{\phi}F^s\right\},\label{HJB_seller} \\
  \ds F^s(T,S,v)=R\left(Z(S)-v\right),
\end{array}\right.
\eeq
where
\begin{equation}
  \mathcal{L}_1^{\phi}F=\frac{1}{2}S^2\sigma^2\frac{\partial^2 F}{\partial S^2}+\phi S^2\sigma^2 \frac{\partial^2 F}{\partial S\partial v}+\frac{1}{2}S^2\sigma^2 \phi^2\frac{\partial^2 F}{\partial v^2}+\mu S\frac{\partial F}{\partial S}+(rv+(\mu-r)\phi S)\frac{\partial F}{\partial v}.
\end{equation}
The minimum risk exposure for the seller is then given by $\rho^s(S,v;Z)=F^s(0,S,v)$.

A similar  analysis can be applied to the buyer. Assume the buyer pays $v$ for the European contingent claim $Z(S_T)$. To finance his payment, the buyer  borrows $v$ at time 0,  which results in a liability of $ve^{rT}$ at expiry. The buyer  establishes a hedging account with a hedging strategy $\phi(\cdot)$ with the zero initial value.  Then the minimum risk exposure for the buyer is defined by
\begin{equation}\label{buyer_risk}
  \rho^b(S,v;Z)= \inf\limits_{\phi(\cdot)\in \mathbf{\Phi}}  \mathbf{E}_{\mathbb{P}}^{v,S}\left[R\left(ve^{rT}-v_T^{0,\phi(\cdot)}-Z(S_T)\right)\right]=\inf\limits_{\phi(\cdot)\in \mathbf{\Phi}}  \mathbf{E}_{\mathbb{P}}^{v,S}\left[R\left(v_T^{v,-\phi(\cdot)}-Z(S_T)\right)\right].
\end{equation}
The associate  HJB equation governing the value function $F^b(t,S,v)$ is given by
\beq
\left\{
\begin{array}{l}
\ds  0=\frac{\partial F^b}{\partial t}+\inf\limits_{\phi\geq 0}\left\{\mathcal{L}_2^{\phi}F^b\right\},\label{HJB_buyer} \\
  \ds F^b(T,S,v)=R\left(v-Z(S)\right),
\end{array}\right.
\eeq
where
\begin{equation}
  \mathcal{L}_2^{\phi}F=\frac{1}{2}S^2\sigma^2\frac{\partial^2 F}{\partial S^2}-\phi S^2\sigma^2 \frac{\partial^2 F}{\partial S\partial v}+\frac{1}{2}S^2\sigma^2 \phi^2\frac{\partial^2 F}{\partial v^2}+\mu  S\frac{\partial F}{\partial S}+(rv+(\mu-r)\phi S)\frac{\partial F}{\partial v}.
\end{equation}
The minimum risk exposure for the buyer is then given by $\rho^b(S,v;Z)=F^b(0,S,v)$.

In order to ensure that the  optimal control problems (\ref{seller_risk}) and (\ref{buyer_risk}) are well-posed, some  conditions must be imposed on the utility function $R(x)$ and the admissible set $\Phi$.    In this paper, we make the following assumption\footnote{ Readers who are interested in these conditions,  are referred to \cite{fleming2006controlled,ma2019optimal}.}.
\begin{assumption}\label{assumption1}
Given the risk function $R(x)$ and the payoff  $Z(S)$, there exists an admissible strategy $\phi(\cdot)$ such that $R\left(Z(S_T)-v_T^{v,\phi(\cdot)}\right)$ and $R\left(v_T^{v,-\phi(\cdot)}-Z(S_T)\right)$ are square integrable.
\end{assumption}

In fact, there exists a direct relationship between the  risk exposures for  the seller and the  buyer. In terms of financial interpretation, a buyer who purchases a European contingent claim $Z(S)$ at a price $v$ is equivalent to a seller who sells a contingent claim $-Z(S)$ at the price of $-v$ because they have identical cash flows. Therefore, they should face the same risks. Mathematically, it is expressed as
\begin{equation}\label{equivalence}
  \rho^b(S,v;Z)=\rho^s(S,-v;-Z).
\end{equation}
This relation plays an important role in the rest of this paper.

The following lemma provides some useful properties of the risk function.
\begin{lemma}\label{lemma1} Assume that $Z,Z_1,Z_2$ are square integrable, $\mathcal{F}_T$-measurable random variables. The monotonicity and limiting behavior of risk functions $\rho^s(S,v;Z)$ and $\rho^b(S,v;Z)$ are described as follows:
\begin{enumerate}
  \item
  If $Z_1\leq Z_2$, then $\rho^s(S,v;Z_1)\leq\rho^s(S,v;Z_2)$ and $\rho^b(S,v;Z_1)\geq\rho^b(S,v;Z_2)$.\\
  If $v_1\leq v_2$, then $\rho^s(S,v_1;Z)\geq\rho^s(S,v_2;Z)$ and $\rho^b(S,v_1;Z)\leq\rho^b(S,v_2;Z)$.
  \item As $v$ tends toward $\i$ or $-\i$, the asymptotic behavior of the risk functions are given by
\begin{align}
\lim\limits_{v\rightarrow\i}\rho^s(S,v;Z)&=L^B,\lim\limits_{v\rightarrow\i}\rho^b(S,v;Z)=\i,\nn\\
\lim\limits_{v\rightarrow -\i}\rho^s(S,v;Z)&=\i,\lim\limits_{v\rightarrow -\i}\rho^b(S,v;Z)=L^B\nn,
\end{align}
where $L^B$  represents the lower bound of the utility function $R(x)$.
\end{enumerate}
\end{lemma}
\begin{proof}
The proof of Lemma \ref{lemma1} is given  in Appendix \ref{appendix_1}.
\end{proof}
We adopt   the definition of the \textit{equal-risk price} provided by \cite{guo2017equal}.
\begin{definition}\label{def_equal}  Consider a European contingent claim with the payoff   $Z(S_T)$.
The equal-risk price of this claim, denoted by $\bar{v}(S)$ where $S$ is the time 0 value of the underlying stock, is  a constant under which both the seller and the buyer face the same amount of risk, i.e.
\begin{equation}\label{equal_risk}
 \rho^s(S,\bar{v}(S);Z)=\rho^b(S,\bar{v}(S);Z).
\end{equation}
\end{definition}
In order to demonstrate that the equal-risk price is well-defined, the following theorem states its existence and uniqueness.
\begin{theorem}\label{thm1}
Consider a market where the stock follows the Black-Scholes model and short selling is banned. For a European contingent claim $Z(S_T)$, there exists a unique equal-risk price $\bar{v}(S)$ that satisfies the following equation,
\begin{equation}\label{equal_risk1}
 \rho^s(S,\bar{v}(S);Z)=\rho^b(S,\bar{v}(S);Z).
\end{equation}
\end{theorem}
\begin{proof}
The proof of this theorem is given in Appendix \ref{appendix_2}.
\end{proof}

In summary, the equal-risk pricing approach (ERP) consists of two steps. First, we calculate the minimum risk exposure for the seller and the buyer respectively through solving the associate stochastic optimal control problems.  In the second step, the equal-risk price is found by solving (\ref{equal_risk}). Mathematically, the idea of equal-risk pricing approach seems like the classical exponential hedging problem: i.e. the superreplication under Gamma constraints, or under delta constraints approach \citep{Bouchard2016Hedging}. In nature, they are two different problem: the former focuses on pricing; while the latter aims at hedging risk. To solve the associate stochastic control problems, we need to solve the HJB equations (\ref{HJB_seller}) and (\ref{HJB_buyer}). In some special cases, these HJB equations can be solved analytically and the pricing formula for the equal-risk price can be derived easily. However, for general claims, analytical solutions  is unavailable and hence we will provide the corresponding numerical schemes.

\section{Equal-risk prices under the equivalent martingale measure}
In this section, we focus on the case where the expected risks are computed under the equivalent martingale measure $\mathbb{Q}$.
We will recover the analytical pricing formulae for European call and put options under the short selling ban from \cite{guo2017equal} and compare them to Black-Scholes prices. Moreover, we will further extend the results of \cite{guo2017equal} to payoffs that are not necessarily monotonic, and provide an analytical formula for general contingent claims.

Since we are working under the martingale measure $\mathbb{Q}$, the parameter $\mu$ in both (\ref{HJB_seller}) and (\ref{HJB_buyer}) is replaced by $r$. %It can be verified that Assumption \ref{assumption1}  holds for the payoff  $Z(S)=(S-K)^+$ and the risk functions $R_1(x)$ and $R_2(x)$.
We derive the minimum risk exposure for the seller and the buyer as follows.
\begin{proposition}\label{proposition1}
When the contingent claim is a European call option with the payoff  $Z(S)=(S-K)^+$, the minimum risk exposure for the seller is
\begin{equation}\label{seller_risk_value}
 \rho^s(S,v;Z) = R\left(e^{rT}[C^{BS}(S,K,r,\sigma,T)-v]\right),
\end{equation}
where $C^{BS}(S,K,r,\sigma,T)$ is the  Black-Scholes formula for a European call option with underlying price $S$, strike price $K$, and time to expiration $T-t$.
\end{proposition}
\begin{proof}
In order to derive the minimum risk exposure  for the seller, we  focus on the HJB equation (\ref{HJB_seller}) with $Z=(S-K)^+$. Consider the following trial solution to the PDE system (\ref{HJB_seller}),
\begin{equation}\label{trial}
  F^s(t,S,v)=R\left(e^{r(T-t)}[C^{BS}(S,K,r,\sigma,T-t)-v]\right).
\end{equation}
It is easy to verify that
\begin{displaymath}
\begin{array}{llllll}
   \frac{\partial F^s}{\partial t}  & = &  e^{r(T-t)}(\frac{\partial C^{BS}}{\partial t}-rC^{BS}+rv) D^+R(x), & \frac{\partial F^s}{\partial S} &= &  e^{r(T-t)}\frac{\partial C^{BS}}{\partial S}D^+R(x), \\
\frac{\partial F^s}{\partial v} & = &  -e^{r(T-t)}D^+R(x), &  \frac{\partial^2 F^s}{\partial S\partial v} &= & -e^{2r(T-t)} \frac{\partial C^{BS}}{\partial S}D^{2+}R(x), \\
 \frac{\partial^2 F^s}{\partial S^2}& = & e^{2r(T-t)}(\frac{\partial C^{BS}}{\partial S})^2D^{2+}R(x)+e^{r(T-t)}\frac{\partial^2 C^{BS}}{\partial S^2}D^+R(x), &   \frac{\partial^2 F^s}{\partial v^2} &= & e^{2r(T-t)}D^{2+}R(x),
\end{array}
\end{displaymath}
where $D^+R(x)$ and $(D^{2+})R(x)$ represent the first  and second order right derivatives. Based on the convexity of function $\mathcal{L}_1F$ with respect to $\phi$, the optimal hedging strategy is
\begin{equation}\label{optimal_phi1}
  \phi^*=\max\left\{-\frac{\partial^2 F^s}{\partial S\partial v}/ (\frac{\partial^2 F^s}{\partial v^2}),0\right\}=\max\left\{\frac{\partial C^{BS}}{\partial S},0\right\}.
\end{equation}
The Delta of a European call option $\ds\frac{\partial C^{BS}}{\partial S}$ is non-negative, which implies $\ds \phi^*=\frac{\partial C^{BS}}{\partial S}$.
After substituting $\phi^*$ back into the HJB equation (\ref{HJB_seller}), we have
\begin{align}
   &\frac{\partial F^s}{\partial t}+\inf\limits_{\phi \geq 0}\left\{\frac{1}{2}\sigma^2S^2\frac{\partial^2 F^s}{\partial S^2}+\phi S^2\sigma^2 \frac{\partial^2 F^s}{\partial S \partial v}+\frac{1}{2}S^2\sigma^2 \phi^2\frac{\partial^2 F^s}{\partial v^2}+rS\frac{\partial F^s}{\partial S}+rv\frac{\partial F^s}{\partial v}\right\}\nn\\
&=e^{r(T-t)}\left[\frac{\partial C^{BS}}{\partial S}+\frac{1}{2}\sigma^2S^2\frac{\partial^2 C^{BS}}{\partial S^2}+rS\frac{\partial C^{BS}}{\partial S}-rC^{BS}\right]\frac{\partial R}{\partial x},\nn\\
&=0\label{demo1}.
\end{align}
The last equation holds just because  $C^{BS}$ satisfies the Black-Scholes PDE.
Consequently, the trial solution (\ref{trial}) is indeed a solution to the HJB equation (\ref{HJB_seller}). Therefore, the minimum risk exposure for the seller can be expressed by (\ref{seller_risk_value}).
\end{proof}
\begin{remark}\label{remark3}
The seller of a European call option  adopts the same optimal hedging strategy as the classical Black-Scholes model, i.e. $\phi^*=\frac{\partial C^{BS}}{\partial S}$, which implies  that the short selling ban  does not affect his hedging strategy. This follows form the fact that the price of a European call option in the classical Black-Scholes model is non-decreasing with respect to the underlying.
\end{remark}
\begin{proposition}\label{proposition2}
When the contingent claim is a European call option with the payoff  $Z(S)=(S-K)^+$, the  minimum risk exposure for the buyer is
\begin{equation}\label{buyer_risk_value1}
  \rho^b(S,v;Z)=\frac{1}{\sqrt{2\pi}}\int_{-\i}^{\i}R\left(ve^{rT}-(Se^{(r-\frac{\sigma^2}{2})T+\sigma\sqrt{T}x}-K)^+\right)e^{-\frac{x^2}{2}}dx.
\end{equation}
\end{proposition}
\begin{proof}
We first claim that the optimal hedging strategy $\phi^*$ for the buyer should be zero when $Z(S)=(S-K)^+$, i.e.
\begin{equation}\label{opti1}
  \rho^b(S,v;Z)=\mathbf{E}_{\mathbb{Q}}R(v_T^{v,0}-Z).%=\mathbf{E}_{\mathbb{Q}}R(ve^{rT}-Z).
\end{equation}
It suffices to demonstrate that $\mathbf{E}_{\mathbb{Q}}R(v_T^{v,-\phi(\cdot)}-Z)\geq\mathbf{E}_{\mathbb{Q}}R(ve^{rT}-Z)$ for any $\phi(\cdot)\in\Phi$. According to the dynamics (\ref{hedging_account1}), we have $v_t^{v,-\phi(\cdot)}=ve^{rt}-\sigma\int_0^te^{r(t-u)}\phi_uS_udW_u$. Since $R(x)$ is a convex function, we have
\begin{equation}\label{c1}
  \mathbf{E}_{\mathbb{Q}}\left[R(v_T^{v,-\phi(\cdot)}-Z)-R(ve^{rT}-Z)\right]\geq \mathbf{E}_{\mathbb{Q}}\left[-D^+R(ve^{rT}-Z)\sigma\int_0^Te^{r(T-u)}\phi_uS_udW_u\right].
\end{equation}
Following Lemma 3.2 in \cite{guo2017equal}, the random variable $-D^+R\left(ve^{rT}-Z\right)$ can be expressed as
\begin{equation}\label{c2}
-D^+R\left(ve^{rT}-Z\right)= -\mathbf{E}_{\mathbb{Q}}\left[D^+R\left(ve^{rT}-Z\right)\right]+\int_0^T\psi_u\sigma S_udW_u,
\end{equation}
where $\psi(\cdot)$ is non-negative. By It\^o isometry, we have
\begin{equation}\label{c3}
\mathbf{E}_{\mathbb{Q}}\left[-D^+R(ve^{rT}-Z)\sigma\int_0^Te^{r(T-u)}\phi_uS_udW_u\right]= \mathbf{E}_{\mathbb{Q}}\int_0^T\sigma^2 e^{r(T-u)}\phi_u\psi_uS_u^2du\geq 0,
\end{equation}
which completes the proof for  our claim (\ref{opti1}).  On the other hand, it yields that from  (\ref{opti1})
\begin{equation}\label{reveiwe_2}
  \rho^b(S,v;Z)=\mathbf{E}_{\mathbb{Q}}[R(v_T^{v,0}-(S_T-K)^+)]=\mathbf{E}_{\mathbb{Q}}[R(ve^{rT}-(Se^{(r-\frac{\sigma^2}{2})T+\sigma W_T}-K)^+)],
\end{equation}
where $W_T\sim N(0,T)$ under the measure $\mathbb{Q}$.
%Since the optimal trading strategy $\phi^*$ is  zero, the HJB equation (\ref{HJB_buyer}) becomes
%\beq
%\left\{
%\begin{array}{l}
%\ds  0=\frac{\partial F^b}{\partial t}+\frac{1}{2}\sigma^2S^2\frac{\partial^2 F^b}{\partial S^2}+rS\frac{\partial F^b}{\partial S}+rv\frac{\partial F^b}{\partial v}.  \label{simplied_HJB_buyer1} \\
%  \ds F^b(T,S,v)=R\left(v-Z(S)\right).
%\end{array}\right.
%\eeq
%By introducing time reversal $\t=T-t$ and function $G(\t,S,v)=F^b(t,S,v)$, we have
%\beq
%\left\{
%\begin{array}{l}
%\ds  \frac{\partial G}{\partial \t}=\frac{1}{2}\sigma^2S^2\frac{\partial^2 G}{\partial S^2}+rS\frac{\partial G}{\partial S}+rv\frac{\partial G}{\partial v},  \label{simplied_HJB_buyer2} \\
%  \ds G(0,S,v)=R\left(v-Z(S)\right).
%\end{array}\right.
%\eeq
%According to the Feynman-Kac formula, the solution of this linear PDE system can be written as the following condition expectation
%\begin{align}
%  G(\t,S,v) &=\mathbf{E}_{\mathbb{Q}}^{v,S}R\left(v_{\t}^{v,0}-(S_{\t}-K)^+\right)\nn\\
%  &=\frac{1}{\sqrt{2\pi}}\int_{-\i}^{\i}R\left(ve^{r\t}-(Se^{(r-\frac{\sigma^2}{2})\t+\sigma\sqrt{\t}x}-K)^+\right)e^{-\frac{x^2}{2}}dx.
%\end{align}
After some simple calculations, it is easy to check the right hand of Equation (\ref{reveiwe_2}) is exactly the same with the right hand of Equation (\ref{opti1}), which completes the proof.
\end{proof}
\begin{remark}\label{remark4}
The optimal hedging strategy for the buyer of the European call option is to hold no stocks when the short selling is banned, which is completely different from the classical Black-Scholes model. The reason is that the optimal hedging strategy in the  Black-Scholes model $\phi^*=-\frac{\partial C^{BS}}{\partial S}$ is non-positive, which is infeasible due to the short selling ban.
\end{remark}
After deriving the minimum  risk exposure for the seller and the buyer, the analytical equal-risk price for the European call option is provided in the following theorem.
\begin{theorem}\label{thm2}
When the short selling is banned in the Black-Scholes model, the equal-risk price of the European call option is given as follows.
\begin{enumerate}
  \item When the risk function is $R_1(x)=x^+$,  the equal-risk price $v$ is given by
\begin{equation}\label{call1}
  v=C^{BS}(S,K,r,\sigma,T)-\left[P^{BS}(S,K+ve^{rT},r,\sigma,T)-P^{BS}(S,K,r,\sigma,T)\right],
\end{equation}
where $P^{BS}(S,K,r,\sigma,T)$ is the  Black-Scholes formula for a European put option .
  \item When the risk function is $R_2(x)=e^x-1$, the equal-risk price $v$ is explicitly expressed as
  \begin{equation}\label{equal_risk_price_call}
  v=\frac{1}{2}\left\{C^{BS}(S,K,r,\sigma,T)-e^{-rT}\ln{\left(\frac{1}{\sqrt{2\pi}}\int_{-\i}^{\i}e^{-(Se^{(r-\frac{\sigma^2}{2})T+\sigma\sqrt{T}x}-K)^+-\frac{x^2}{2}}dx\right)}\right\}.
\end{equation}
\end{enumerate}
\end{theorem}
\begin{proof}
The minimum risk exposure for the  seller and the buyer have been derived in Propositions \ref{proposition1} and \ref{proposition2}. According to Definition \ref{def_equal}, the equal-risk price of the European call option  is the root of the following  equation,
\begin{equation}\label{equal_risk_equation}
  \rho^s\left(S,v;(S-K)^+\right)=\rho^b\left(S,v;(S-K)^+\right).
\end{equation}
It is then straightforward to verify that the equal-risk prices are given by (\ref{call1}) and (\ref{equal_risk_price_call}).
\end{proof}
\begin{remark}
Note that the analytical pricing formula (\ref{call1}) is same with those provided by \cite{guo2017equal} when the risk function is $R_1(x)=x^+$, which demonstrates that our PDE approach is consistent with the results from \cite{guo2017equal}. In addition, we derive another explicit pricing formula in (\ref{equal_risk_price_call}) for the case when the risk function is given by $R_2(x)=e^x-1$.  The main difference between these  formulae (\ref{call1}) and (\ref{equal_risk_price_call}) is that  the former  is not explicit and it must be solved by root finding algorithms, while  the latter  is explicit.
\end{remark}
The pricing formula (\ref{call1}) was interpreted as the standard Black-Scholes price with an adjustment term in \cite{guo2017equal}. In this paper, we mainly focus on the new explicit equal-risk price (\ref{equal_risk_price_call}) when the risk function is  $R_2(x)=e^x-1$.
To illustrate how the short selling ban affects the European call option price, we compare the results computed from the pricing formula (\ref{equal_risk_price_call}) with those calculated from the standard Black-Scholes formula in Figure \ref{figure1a} under the following parameters
\begin{equation}\label{par}
  K=10, r=0.05,T=0.5,\sigma=0.3.
\end{equation}
As shown in Figure \ref{figure1a}, the absolute differences between the equal-risk prices and the Black-Scholes prices are significant for large underlying prices, which indicates that the short selling ban affects the value of European call option substantially. When the underlying stock price is low, the price of a call option is also low regardless of the short selling ban. As a result, the absolute differences between the equal-risk prices and the Black-Scholes prices are not very significant. In order to demonstrate the effect for small underlying prices,  we define the relative difference with respect to Black-Scholes prices as
\begin{equation}\label{perctentage}
 \frac{\text{Equal-risk price}-\text{Black-Scholes price}}{\text{Black-Scholes price}}\times 100\%,
\end{equation}
which is depicted in Figure \ref{figure1b}. It is also observed that the relative difference is substantial even for small underlying prices. From Figures \ref{figure1a} and \ref{figure1b}, we conclude that the short selling ban  significantly lowers the value of the European call option.

\begin{figure}[!htbp]
\centering \subfigure[Two prices for the European call option]{
\label{figure1a}
\includegraphics[width=0.475\textwidth]{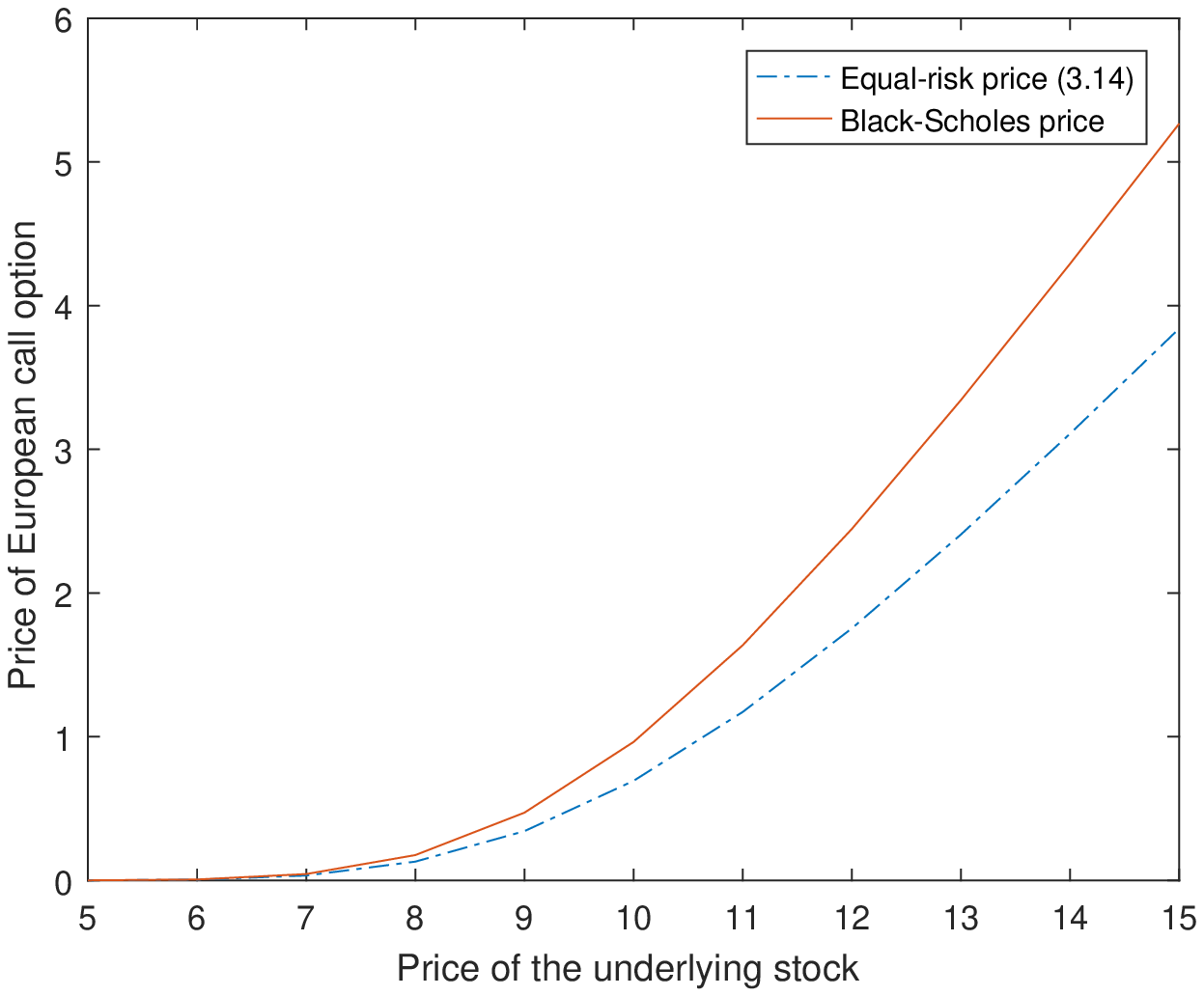}
} \subfigure[Relative difference] {
\label{figure1b}
\includegraphics[width=0.475\textwidth]{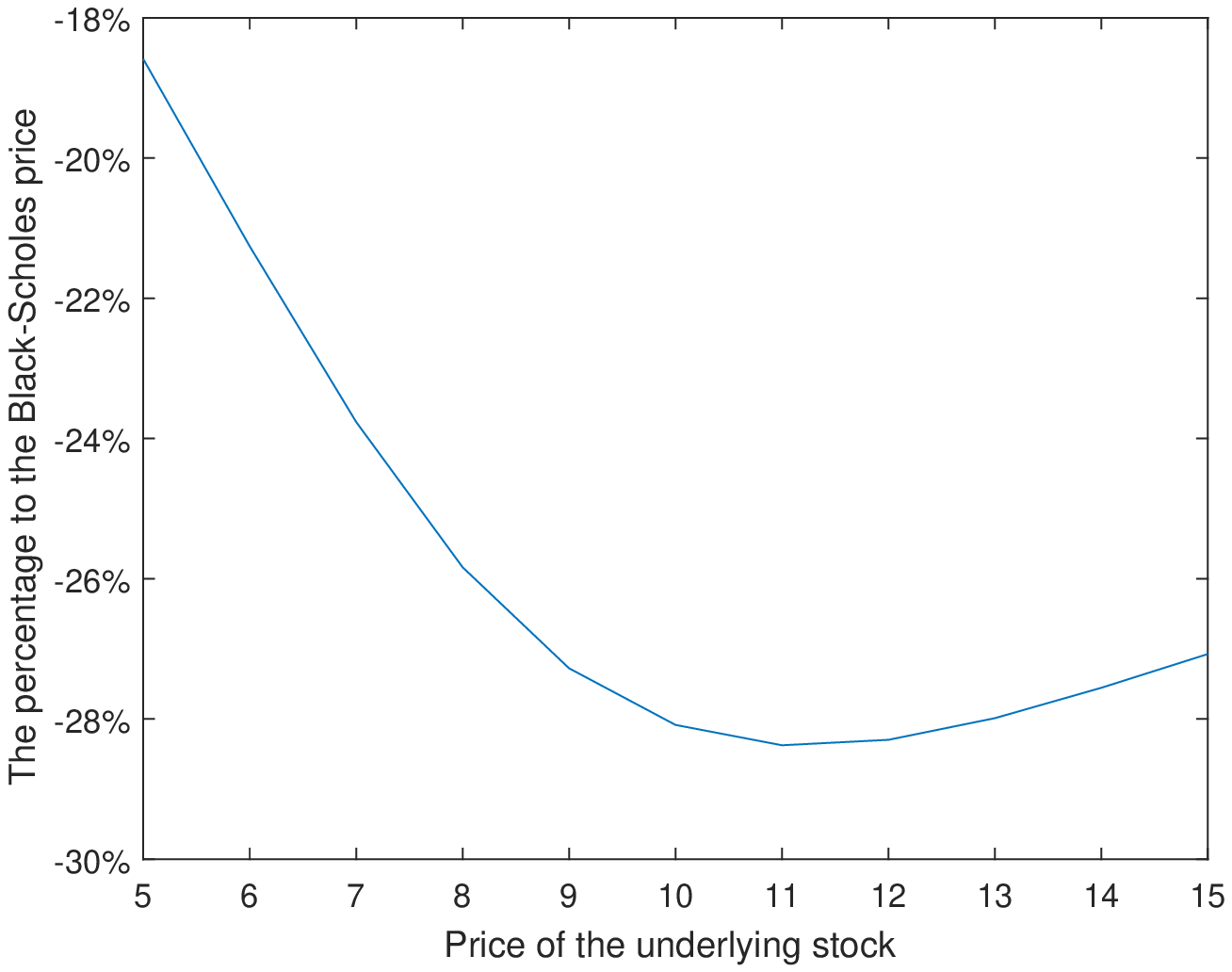}
}
\caption{Comparisons between the equal-risk prices and the Black-Scholes prices.}
\end{figure}

From Propositions \ref{proposition1} and \ref{proposition2}, the optimal hedge in the  Black-Scholes model is still available for the seller of the European call option, while the counterpart is unavailable for the buyer due to the short selling ban. If the transaction price of contingent claim is still set to be Black-Scholes price, the seller  faces no risk, but the buyer incurs substantial risks because the optimal hedging strategy is no longer available.  The equal-risk price redistributes the risk between the buyer and seller equally, transferring some risk from the buyer to the seller. As a result, the equal-risk price should be lower than  the Black-Scholes  price. This is also seen from (\ref{equal_risk_price_call}). Indeed,  ((\ref{equal_risk_price_call})) and Jensen's inequality provide
\begin{eqnarray*}
  2(C^{BS}-v) &=& C^{BS}+e^{-rT}\log \mathbf{E}_{\mathbb{Q}}\left[e^{-(S_T-K)^+}\right] \\
   &=& e^{-rT}\left(\log \mathbf{E}_{\mathbb{Q}}\left[e^{-(S_T-K)^+}\right]-   \mathbf{E}_{\mathbb{Q}}\left[e^{-(S_T-K)^+}\right]\right)\\
   &\geq& e^{-rT}\left( \mathbf{E}_{\mathbb{Q}}\left[e^{-(S_T-K)^+}\right]-   \mathbf{E}_{\mathbb{Q}}\left[e^{-(S_T-K)^+}\right]\right)=0.
\end{eqnarray*}

According to the relation (\ref{equivalence}) between the minimum risk exposure for the buyer and the seller, we also derive the equal-risk price for the European put option as corollaries. The proofs are left in Appendix \ref{appendix_3}.
\begin{corollary}\label{coro1}
When the contingent claim is a European put option with payoff   $Z(S)=(K-S)^+$,  the minimum  risk exposure  for the buyer is
\begin{equation}\label{seller_risk_value3}
 \rho^b(S,v;Z) = R\left(e^{rT}[v-P^{BS}(S,K,r,\sigma,T)]\right).
\end{equation}
\end{corollary}
\begin{corollary}\label{coro2}
When the contingent claim is a European put option with payoff   $Z(S)=(K-S)^+$,  the minimum  risk exposure for the seller  is
\begin{equation}\label{buyer_risk_value4}
  \rho^s(S,v;Z)=\frac{1}{\sqrt{2\pi}}\int_{-\i}^{\i}R\left((K-Se^{(r-\frac{\sigma^2}{2})T+\sigma\sqrt{T}x})^+-ve^{rT}\right)e^{-\frac{x^2}{2}}dx.
\end{equation}
\end{corollary}
\begin{corollary}\label{coro3}
If short selling is banned in the Black-Scholes model, the equal-risk price of the European put option is given as follows.
\begin{enumerate}
  \item Under the risk function $R_1(x)=x^+$, the equal-risk price of the European put option satisfies
 \begin{equation}\label{put_price}
   v=P^{BS}(S,K,r,T,\sigma)+P^{BS}(S,K-ve^{rT},r,T,\sigma).
 \end{equation}
  \item Under the risk function $R_2(x)=e^x-1$, the equal-risk price of the European put option is explicitly expressed as
  \begin{equation}\label{formula_thm2}
  v=\frac{1}{2}\left\{P^{BS}(S,K,r,T,\sigma)+e^{-rT}\ln{\left(\frac{1}{\sqrt{2\pi}}\int_{-\i}^{\i}e^{(K-Se^{(r-\frac{\sigma^2}{2})T+\sigma\sqrt{T}x})^+-\frac{x^2}{2}}dx\right)}\right\}.
\end{equation}
\end{enumerate}
\end{corollary}
When the risk function is $R_1(x)=x^+$, the equal-risk price (\ref{put_price}) for the European put option coincides with the results of \cite{guo2017equal}.  Using the parameters from (\ref{par}), a comparison between the equal-risk price (\ref{formula_thm2}) and the Black-Scholes price is plotted in Figure \ref{figure2a} to demonstrate the effect of the short selling ban on European put options.
\begin{figure}[!htbp]
\centering \subfigure[Two prices for European put option.]{
\label{figure2a}
\includegraphics[width=0.475\textwidth]{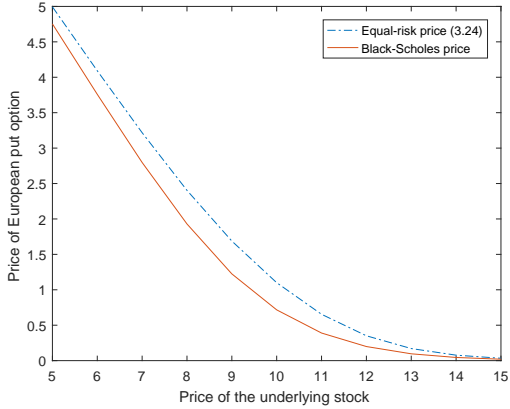}
} \subfigure[Relative difference] {
\label{figure2b}
\includegraphics[width=0.475\textwidth]{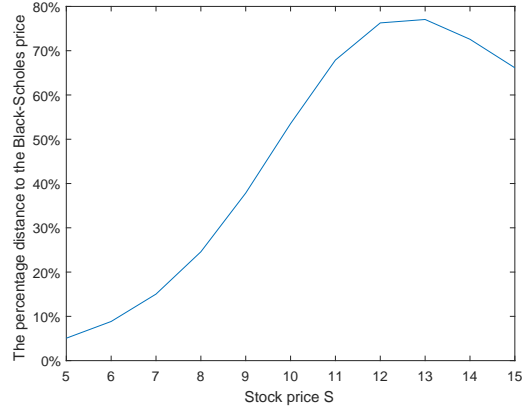}
}
\caption{Comparisons between equal-risk price and Black-Scholes price.}
\end{figure}
From Figure \ref{figure2a}, the absolute difference between the two prices is significant when the underlying price is small. The relative distance of the equal-risk price with respect to the Black-Scholes price is depicted in Figure \ref{figure2b}, which indicates that the relative difference is significant even though the absolute difference is small for large underlying prices. From both Figures \ref{figure2a} and \ref{figure2b}, we conclude that the equal-risk price of a European put option is higher than the Black-Scholes price. Similarly, it can also been seen  from (\ref{formula_thm2}). Indeed,  (\ref{formula_thm2}) and Jensen's inequality provide
\begin{eqnarray*}
  2(v-P^{BS}) &=&
  e^{-rT}\left(\log \mathbf{E}_{\mathbb{Q}}\left[e^{(K-S_T)^+}\right]-   \mathbf{E}_{\mathbb{Q}}\left[e^{(K-S_T)^+}\right]\right)\geq 0.
\end{eqnarray*}
 In other words, the short selling ban has increased the European put option price substantially. Compared with the  Black-Scholes model, the buyer  pays more to purchase a European put option when short selling is banned.   This difference represents the additional risk the buyer takes under the equal-risk framework.

In summary, we have explicitly computed the equal-risk prices of  European call and put options by solving the HJB equations (\ref{HJB_seller}) and (\ref{HJB_buyer}). The pricing formula  is consistent with the results derived by \cite{guo2017equal}. However, recall that the analytical results of \cite{guo2017equal} are only restricted to  monotonic payoffs.
We will extend them to the case of non-monotonic payoffs, such as a butterfly spread option, by providing a generalization of Propositions \ref{proposition1}  and \ref{proposition2}.

Let $Z$ be an $\mathcal{F}$-measurable random variable satisfying $\mathbf{E}_{\mathbb{Q}}[Z^2]<\i$.  Then there exists an $\mathcal{F}$-progressively measurable process $\pi$ such that $Z=v^{z,\pi}_T$ and $\mathbf{E}_{\mathbb{Q}}[\int_0^T\pi(t,\omega)^2dt]<\i$. In other words, $\pi$ is a perfect hedging strategy for $Z$ if there are no trading restrictions, and $z=\mathbf{E}_{\mathbb{Q}}[e^{-rT}Z]$ is the price of the claim under the Black-Scholes model.
\begin{proposition}\label{propnonmonotonic}
	Consider a contingent claim with payoff $Z$ and short selling of the underlying is banned. Then, the minimum risk exposures for the seller and the buyers under the equivalent martingale measure $\mathbb{Q}$ are given by
	\begin{equation}\label{ee}
		\rho^s(S,v;Z)=\mathbf{E}_{\mathbb{Q}}[R(Z-v_T^{v,\phi^{s*}})], \quad \quad \rho^b(S,v;Z)=\mathbf{E}_{\mathbb{Q}}[R(v_T^{v,-\phi^{b*}}-Z)],
	\end{equation}
where $\phi^{s*}=\pi^+\in \Phi$ and $\phi^{b*}=\pi^-\in \Phi$.
 In other words, the optimal hedging strategies involve mimicking the Black-Scholes hedge whenever possible, except holding zero units of the underlying if a short position is required.
\end{proposition}
\begin{proof}
Note that $\phi^{s*}$ and $\phi^{b*}$ are the positive and negative parts of $\pi$.
Denote $\hat{v}_t^{v,\phi}=e^{-rt}v_t^{v,\phi}$ and $\hat{S}_t=e^{-rt}S_t$. We then notice that, under $\mathbb{Q}$,
  \begin{eqnarray*}
    \langle \hat{v}^{0,\pi-\phi}\rangle_T &=& \int_0^T(\pi_t-\phi_t)^2\sigma^2\hat{S}_t^2dt\geq \int_0^T(-\pi_t^-)^2\sigma^2\hat{S}_t^2dt=\langle \hat{v}^{0,\pi-\phi^{s*}}\rangle_T, \\
     \langle \hat{v}^{0,\pi+\phi}\rangle_T &=& \int_0^T(\pi_t+\phi_t)^2\sigma^2\hat{S}_t^2dt\geq \int_0^T(\pi_t^+)^2\sigma^2\hat{S}_t^2dt=\langle \hat{v}^{0,\pi+\phi^{s*}}\rangle_T,
     \end{eqnarray*}
 for all $\phi\in \Phi$,     where  $\langle f \rangle_t$ denotes the quadratic variation at time $t$.
These quadratic variation inequalities can then be used to infer inequalities in expected risks via a suitable time change.
Applying the Dambis-Dubins-Schwarz Theorem (See Theorem 3.4.6 and Problem 3.4.7 in \cite{karatzas05}) and Jensen's inequality, we have
       \begin{eqnarray*}
    \mathbf{E}_{\mathbb{Q}}[R(Z-v_T^{v,\phi})] & \geq&   \mathbf{E}_{\mathbb{Q}}[R(e^{rT}[z-v+\hat{v}_T^{0,\pi-\phi}])]=\mathbf{E}_{\mathbb{Q}}[R(e^{rT}[z-v+B_{\langle\hat{v}^{0,\pi-\phi}\rangle_T}])] \\
    &\geq& \mathbf{E}_{\mathbb{Q}}[R(e^{rT}[z-v+B_{\langle\hat{v}^{0,\pi-\phi^{s*}}\rangle_T}])]=  \mathbf{E}_{\mathbb{Q}}[R(Z-v_T^{v,\phi^{s*}})]\\
    \mathbf{E}_{\mathbb{Q}}[R(v_T^{v,-\phi}-Z)] & \geq&    \mathbf{E}_{\mathbb{Q}}[R(e^{rT}[v-z+\hat{v}_T^{0,-(\pi+\phi)}])]=\mathbf{E}_{\mathbb{Q}}[R(e^{rT}[v-z+B_{\langle\hat{v}^{0,-(\pi+\phi)}\rangle_T}])] \\
    &\geq& \mathbf{E}_{\mathbb{Q}}[R(e^{rT}[v-z+B_{\langle\hat{v}^{0,-(\pi+\phi^{b*})}\rangle_T}])]=  \mathbf{E}_{\mathbb{Q}}[R(v_T^{v,-\phi^{b*}}-Z)]
        \end{eqnarray*}
        for all $\phi\in\Phi$, where $B_t$ is a standard one-dimensional Brownian motion under $\mathbb{Q}$. Therefore the expected risks for the seller and buyer are minimised when the hedging strategies are $\phi^{s*}$ and $-\phi^{b*}$, respectively.
\end{proof}
\begin{corollary}
The equal risk price of $Z$ under the risk neutral measure is given by $v$, which is the unique solution to the equation
\begin{equation}
\mathbf{E}_{\mathbb{Q}}[R(Z-v_T^{v,\phi^{s*}})]=\mathbf{E}_{\mathbb{Q}}[R(v_T^{v,-\phi^{b*}}-Z)].
\end{equation}
\end{corollary}
\begin{remark}
Proposition \ref{propnonmonotonic} can be further extended beyond short selling bans to arbitrary convex trading constraints. Instead of $\phi^{s*}=\pi^+$ and $\phi^{b*}=\pi^-$, similar arguments imply that the optimal trading strategies can be found by projecting $\pi$ and $-\pi$ pointwise to set of admissible trading strategies.
\end{remark}
Therefore we have obtained analytical results to compute the equal risk prices of arbitrary contingent claims under the equivalent martingale measure $\mathbb{Q}$. Nevertheless, in terms of physical interpretation, it is preferable to compute expected risks under the physical measure $\mathbb{P}$. We will demonstrate how this can be done in the next section.

%In addition, when the We apply a numerical scheme to calculate the equal-risk price of a butterfly spread option in the next section.

\section{Numerical scheme for equal risk prices under the physical measure}
In this section, we provide a numerical scheme to solve the  HJB equations (\ref{HJB_seller}) and (\ref{HJB_buyer}), where the risk is measured under the original physical measure $\mathbb{P}$.  Mathematically, they are of the same type and we will only focus on the former in our demonstration. The latter can be solved numerically in a similar fashion.
\subsection{Discretization}
We first introduce the time reversal $\t=T-t$, which transforms the HJB equation (\ref{HJB_seller}) into
\beq
\left\{
\begin{array}{l}
   F^s_{\t}=\inf\limits_{\phi\geq 0}\left\{\frac{1}{2}S^2\sigma^2F^s_{SS}+\phi S^2\sigma^2 F^s_{Sv}+\frac{1}{2}S^2\sigma^2 \phi^2F^s_{vv}+\mu SF^s_S+(rv+(\mu-r)\phi S)F^s_v\right\},\\
   F^s(0,S,v)=R\left(Z(S)-v\right), (\t,S,v)\in \Omega:=[0,T]\times[0,\i)\times\mathbf{R}\label{ex1}.
\end{array}\right.
\eeq
Then we truncate the unbounded domain into a bounded one:
\begin{equation*}
 \bar{\Omega}= [0,T]\times[0,S_{\max}]\times[-v_{\max},v_{\max}].
\end{equation*}
In order to establish a properly-closed PDE system,  boundary conditions will be imposed in each of our numerical examples. As pointed out by  \cite{barles1997convergence}, truncating domain incurs some approximation errors, which are expected to be arbitrarily small by extending the computational domain. The domain is then discretized by  a set of uniformly distributed grids  as follows,
\begin{eqnarray*}
  S_i &=& (i-1)\cdot \Delta S, i=1,\cdots,N_1,\\
  v_j &=& (j-1)\cdot \Delta v, j=1,\cdots,N_2, \\
  \t_l &=& (l-1)\cdot\Delta \t, l=1,\cdots,M,
\end{eqnarray*}
where $N_1,N_2$ and $M$ are the grid sizes in the $S$,$v$ and $\t$ directions, respectively. The corresponding step sizes are   $\ds\Delta S=\frac{S_{\max}}{N_1-1}$, $\ds\Delta v=\frac{v_{\max}}{N_2-1}$, and $\ds\Delta \t=\frac{T}{M-1}$. The value of the unknown function $F^s(\t,S,v)$ at a grid point is denoted by $F_{i,j}^n=F^s(\t_n,S_i,v_j).$

Now we  adopt an explicit scheme to approximate the  function $\phi$  as follows:
\begin{equation}
  \phi_{i,j}^n:=\phi(\t_n,S_i,v_j)=\max\left\{-\frac{\Delta v}{4 \Delta S}\frac{F^{n}_{i+1,j+1}+F^{n}_{i-1,j-1}-F^{n}_{i+1,j-1}-F^{n}_{i-1,j+1}}{F^{n}_{i,j+1}-2F^{n}_{i,j}+F^{n}_{i,j-1}},0\right\},
\end{equation}
and then apply an implicit scheme for the function $F$
\begin{equation}
  \frac{F^{n+1}_{i,j}-F^{n}_{i,j}}{\Delta \t}=\mathcal{L}^{\phi_{i,j}^{n}}_3F_{i,j}^{n+1},
\end{equation}
where
\begin{equation}\label{ll}
  \mathcal{L}^{\phi}_3F=aF_{SS}+\rho(\phi) F_{Sv}+b(\phi)F_{vv}+cF_{S}+dF_v,
\end{equation}
with $a=\frac{1}{2}\sigma^2S^2, b(\phi)=\frac{1}{2}\phi^2\sigma^2S^2,\rho(\phi)=\phi\sigma^2S^2, c=\mu S,d=(rv+(\mu-r)\phi S)$.

The alternative direction implicit (ADI) scheme is then applied to split the linear operator $\mathcal{L}_3$ defined in (\ref{ll}) into two steps. In the first step, only the derivatives with respect to $S$ are evaluated in terms of the unknown values $F^{2n+1}$, while the other derivatives are replaced by the known values $F^{2n}$. The difference equation obtained in the first step is implicit in the $S$-direction and explicit in $v$-direction. The procedure is then repeated at next step with the difference equation implicit in the $v$-direction and explicit in the $S$-direction. The cross derivative is always treated explicitly.  Thus, we have two difference equations:
\begin{eqnarray}
  \frac{F_{i,j}^{2n+1}-F_{i,j}^{2n}}{\Delta \t }&=& a_{i}\frac{F_{i+1,j}^{2n+1}-2F_{i,j}^{2n+1}+F_{i-1,j}^{2n+1}}{\Delta S^2}+c_i\frac{F_{i+1,j}^{2n+1}-F_{i-1,j}^{2n+1}}{2\Delta S} \label{former} \\
   & & + b_{i,j}\frac{F_{i,j+1}^{2n}-2F_{i,j}^{2n}+F_{i,j-1}^{2n}}{\Delta v^2}+d_{i,j}\frac{F_{i,j+1}^{2n}-F_{i,j-1}^{2n}}{2\Delta v}\nn\\
   &&+\rho_{i,j}\frac{F_{i+1,j+1}^{2n}-F_{i-1,j+1}^{2n}-F_{i+1,j-1}^{2n}+F_{i-1,j-1}^{2n}}{4\Delta S\Delta v},\nn\\
  \frac{F_{i,j}^{2n+2}-F_{i,j}^{2n+1}}{\Delta \t }&=&  b_{i,j}\frac{F_{i,j+1}^{2n+2}-2F_{i,j}^{2n+2}+F_{i,j-1}^{2n+2}}{\Delta v^2}+d_{i,j}\frac{F_{i,j+1}^{2n+2}-F_{i,j-1}^{2n+2}}{2\Delta v} \label{latter}\\
   & & + a_i\frac{F_{i+1,j}^{2n+1}-2F_{i,j}^{2n+1}+F_{i-1,j}^{2n+1}}{\Delta S^2}+c_i\frac{F_{i+1,j}^{2n+1}-F_{i-1,j}^{2n+1}}{2\Delta S} \nn\\
    &&+\rho_{i,j} \frac{F_{i+1,j+1}^{2n+1}-F_{i-1,j+1}^{2n+1}-F_{i+1,j-1}^{2n+1}+F_{i-1,j-1}^{2n+1}}{4\Delta S\Delta v}.\nn
\end{eqnarray}
The unknown functions $F_{i,j}^n$ and $\phi_{i,j}^n$ are both derived by solving these difference equations.
\begin{remark}
It has also to be pointed out that we haven't proved that  the above ADI numerical scheme is convergent when it is applied to solve the nonlinear HJB equation. Instead, we   would show the numerical convergence on different grids with $l_2$ error in the following examples. Readers, who feel interested in the proof of numerical methods  for solving HJB equation, are referred to  \cite{kushner1990numerical,forsyth2007numerical,wang2008maximal,mazhukang2020}
\end{remark}

After solving the PDE systems (\ref{HJB_seller}) and (\ref{HJB_buyer}), the minimum risk exposure for  the seller and the buyer are produced numerically on the grids. To compute the equal-risk price for the contingent claims, the root of  (\ref{equal_risk}) is solved numerically, which is similar to determining the optimal exercise price from the values of American put option through the free-boundary condition. We demonstrate how to compute equal-risk prices numerically in the following simple example.

Give a current stock price $S$, which is located between two grid points $S_i$ and $S_{i+1}$, i.e. $S\in (S_i,S_{i+1})$. When the offer price $v$ is larger than the equal-risk price $v(S_i)$, the seller would take less risk for he gets more compensation, i.e
\begin{equation}
  \rho^s(S_i,v;Z)<\rho^b(S_i,v;Z), \quad v>v(S_i).
\end{equation}
On the other hand, when the offer price $v$ is smaller than the equal-risk price $v(S_{i})$, the buyer takes less risk because he pays less, i.e.
\begin{equation}
  \rho^s(S_i,v;Z)>\rho^b(S_i,v;Z), \quad v<v(S_i).
\end{equation}
Consequently, the equal-risk price of the claim $Z$ with current price $S_{i}$ is given by
\begin{equation}\label{S1}
  v(S_i)=\max\limits_{j}\{v_j,j=1,\cdots,N_2 \quad| \quad \rho^s(S_i,v_j;Z)>\rho^b(S_i,v_j;Z)\}.
\end{equation}
Similarly, the equal-risk price of the claim $Z$ with current price $S_{i+1}$ is obtained as
\begin{equation}\label{S2}
  v(S_{i+1})=\max\limits_{j}\{v_j,j=1,\cdots,N_2 \quad|\quad \rho^s(S_{i+1},v_j;Z)>\rho^b(S_{i+1},v_j;Z)\}.
\end{equation}
Finally, the  equal-risk price of  the contingent claim $Z$ with current price $S$ is approximated by
\begin{equation}\label{equal_risk_price_S}
  v(S)=\frac{v(S_i)+v(S_{i+1})}{2}.
\end{equation}

\subsection{Numerical examples}
In this subsection, two numerical examples are provided to illustrate the performance and convergence of our numerical scheme.  Both examples are carried out with Matlab 2016a on an Intel(R) Xeon (R) CPU and the risk function is taken to be $R_2(x)=e^x-1$.
\subsubsection{Example  1: European call option}
In the first example, the contingent claim is a European call option, of which the value functions $F^s(t,S,v)$ and  $F^b(t,S,v)$ have been obtained analytically according to Propositions \ref{proposition1} and \ref{proposition2}. The analytical solutions are considered as the benchmark to illustrate the performance of our numerical scheme.  Before implementing our numerical scheme, we need provide the proper boundary conditions for the PDE systems (\ref{HJB_seller}) and (\ref{HJB_buyer}).

First of all, at the boundary $S=0$, since the stock follows a geometric Brownian motion, the European call option is worthless at expiry. The seller of such a claim faces no liability; while the buyer gets nothing. In addition, the hedging strategies for both the seller and the buyer must be $\phi^*=0$ because they could not invest on a stock whose price is zero.  Therefore, the boundary conditions at $S=0$  are
\beq
\left\{
\begin{array}{l}
F^s(t,0,v)=R(-ve^{r(T-t)}), \\
F^b(t,0,v)=R(ve^{r(T-t)}).\\
\end{array}\right.
\eeq
On the other hand, $S\rightarrow\i$ implies $S_T\rightarrow\i$, which indicates that the European call option is priceless. The buyer of such a claim would have an infinite income at expiry. The boundary condition for the buyer at $S\rightarrow\i$ is imposed as
\begin{equation}\label{buyer_Si}
    \lim\limits_{S\rightarrow\i}F^b(t,S,v)=\lim\limits_{S\rightarrow\i}\inf\limits_{\phi(\cdot)\in\Phi}\mathbf{E}_{\mathbb{Q}}^{S,v}\left[R\left(v_T^{v,\phi(\cdot)}-(S_T-K)^+\right)\right]=\lim\limits_{S\rightarrow\i}R(-S)=-1.
\end{equation}
This bounded Dirichlet boundary condition is approximated by
\begin{equation}\label{buyer_St}
    F^b(t,S_{\max},v)=-1.
\end{equation}
As for the seller, we have
 \begin{equation}\label{seller_Si}
    \lim\limits_{S\rightarrow\i}F^s(t,S,v)=\lim\limits_{S\rightarrow\i}\inf\limits_{\phi(\cdot)\in\Phi}\mathbf{E}_{\mathbb{Q}}^{S,v}\left[R\left((S_T-K)^+-v_T^{v,\phi(\cdot)}\right)\right]=\i.
\end{equation}
When the value function approaches infinity on the boundary, we must perform growth order analysis in order to impose the appropriate boundary condition.  For any admissible hedging strategy $\phi$, by applying the Jensen's inequality to the risk function $R(x)$, we have
\begin{align}
  \mathbf{E}_{\mathbb{Q}}^{S,v}\left[R\left(Z(S_T)-v_T^{v,\phi(\cdot)}\right)\right] &\geq R
  \left(\mathbf{E}_{\mathbb{Q}}^{S,v}\left[Z(S_T)-v_T^{v,\phi(\cdot)}\right]\right)\nn\\
  &=R\left(e^{r(T-t)}(C^{BS}(S,K,r,\sigma,T-t)-v)\right).
\end{align}
Consequently, the asymptotic behavior of the value function $F^s(t,S,v)$ is given by
\begin{equation}\label{tes}
  \lim\limits_{S\rightarrow\i}F^s(t,S,v)\geq \lim\limits_{S\rightarrow\i}R\left(C^{BS}(S,K,r,\sigma,T-t)e^{r(T-t)}-ve^{r(T-t)}\right)\rightarrow\i\quad  \text{for $t\in[0,T]$},
\end{equation}
which means that the growth order of $F^s(t,S,v)$ with respect to $S$ is higher than that of the right hand side for any $t$. On the other hand, at the $t=T$, it follows that
\begin{equation}\label{terminal}
  \lim\limits_{S\rightarrow\i}F^s(T,S,v)=\lim\limits_{S\rightarrow\i}R\left((S-K)^+-v\right),
\end{equation}
which implies that the growth order of $F^s(t,S,v)$ is the same as the right hand side of the above equation at $t=T$. In order to make sure the boundary condition at $S\rightarrow\i$ is consistent with the terminal condition at the corner point, the boundary condition at $S=S_{\max}$ is
\begin{equation}
    F^s(t,S_{\max},v)=R\left((S_{\max}-K)^+-ve^{r(T-t)}\right).
\end{equation}
\begin{remark}
When the value function is bounded, such as Equation (\ref{buyer_Si}), we can directly impose the bound on the truncated boundary, such as Equation (\ref{buyer_St}). When the value function approaches infinity on the boundary, such as Equation (\ref{tes}), we must perform growth order analysis first and then impose an approximate boundary condition similar to Equation (\ref{seller_Si}) to ensure that it is consistent with the terminal condition. For the rest of this paper, we will provide the boundary conditions derived via these steps without providing the full details.
\end{remark}
Following Lemma \ref{lemma1},  the boundary conditions along the $v$ direction are
\beq
\left\{
\begin{array}{l}
\lim\limits_{v\rightarrow\i}F^b(t,S,v)=\i,\\
\lim\limits_{v\rightarrow\i} F^s(t,S,v)=-1, \\
\lim\limits_{v\rightarrow-\i} F^b(t,S,v)=-1,\\
\lim\limits_{v\rightarrow-\i} F^s(t,S,v)=\i. \\
\end{array}\right.
\eeq
which are approximated by
 \beq
\left\{
\begin{array}{l}
F^b(t,S,v_{\max})=R\left(v_{\max}e^{r(T-t)}-(S-K)^+\right),\\
F^s(t,S,v_{\max})=-1,\\
F^b(t,S,-v_{\max})=-1,\\
 F^s(t,S,-v_{\max})=R\left((S-K)^++v_{\max}e^{r(T-t)}\right).\\
 %F^b(t,S,v)=-1.\\
\end{array}\right.
\eeq

After providing these proper boundary conditions for the value functions $F^s(t,S,v)$ and  $F^b(t,S,v)$, we now implement our numerical scheme. The parameters used in the this experiment are listed in Table \ref{parameters}.
   \begin{table}[!htb]
\centering
\begin{tabular}{|c|c|c|c|c|c|c|c|c|}
  \hline
    Parameters & $K$ & $T$ & $\mu$& $r$ & $\sigma$ & $S_{\max} $&$ v_{\max}$&$v_0$  \\
 \hline
  values &5 &0.5 & 0.05 &0.05 & 0.3 &10 & $5$&2  \\
  \hline
\end{tabular}
\caption{Parameters. }
\label{parameters}
\end{table}

Given $\t=T$ and $v=v_0$, the values of $F^s(\t,S,v)$ and $F^b(t,S,v)$ are computed at different values of $S$ and then listed in Tables \ref{table2} and \ref{table3}. To determine the numerical rates of convergence, we choose a sequence of meshes by successively halving the mesh parameters. The analytical solutions (\ref{seller_risk_value}) and (\ref{buyer_risk_value1}) obtained in Propositions \ref{proposition1} and \ref{proposition2} are used as a benchmark when we report the $l_2$ error.  The \textit{ratio} column of Tables \ref{table2} and \ref{table3} is the ratio of successive $l_2$ error as the grid is refined by a factor of two.
\begin{table}[htbp]
\centering
\begin{tabular}{|l|c|l|l|l|l|l|c|c|}
\hline
$(N_1,N_2,M)$&$S=4$ &$S=4.5$&$S=5$ &$S=5.5$&$S=6$&$l_2$ error&ratio\\
\hline
(21,21,160) &  -0.8604  &  -0.8399  &  -0.7981 &   -0.7216 &   -0.5889&0.0452&\\
 (41,41,320)&-0.8595 &   -0.8371  &  -0.7915  &  -0.7074&    -0.5601&0.0123&3.7\\
(81,81,640)& -0.8593&   -0.8364  &  -0.7898  &  -0.7039 &   -0.5528&  0.0040  &3.1 \\
(161,161,1280)& -0.8591   & -0.8361  &  -0.7891  &  -0.7026  &  -0.5503&  0.0012  &3.4 \\
\hline
Benchmark (\ref{seller_risk_value})&  -0.8592  & -0.8362  &  -0.7892 &   -0.7023    &-0.5492&   &\\
\hline
\end{tabular}
\caption{The values of $F^s(T,S,v_0)$  with different meshes.}\label{table2}
\end{table}
\begin{table}[htbp]
\centering
\begin{tabular}{|l|c|l|l|l|l|l|c|c|c|}
\hline
$(N_1,N_2,M)$&$S=4$ &$S=4.5$&$S=5$ &$S=5.5$&$S=6$&$l_2$ error&ratio\\
\hline
(21,21,160) &    6.3860  &  5.7689    &4.8250 &   3.7099   & 2.6162&0.1635&\\
 (41,41,320)&   6.3423 &   5.6985&    4.7540&    3.6598 &   2.5889&0.0403&4.1\\
(81,81,640)&6.3307 &   5.6812 &   4.7369  &  3.6475   & 2.5819&  0.0099 &4.1 \\
(161,161,1280)&6.3302  &  5.6791   & 4.7348   & 3.6465 &   2.5820& 0.0071&1.4 \\
\hline
Benchmark  (\ref{buyer_risk_value1}) &  6.3268  &  5.6755  &  4.7313   & 3.6435    &2.5800&   &\\
\hline
\end{tabular}
\caption{The values of $F^b(T,S,v_0)$ with different meshes.}\label{table3}
\end{table}

From Tables  \ref{table2} and \ref{table3}, it is observed that the successive $l_2$ errors approach zero as the grid spacing is reduced, which show that our numerical results are in good agreement with the benchmark solution.  Therefore, we choose the numerical results calculated on the grid $(161,161,1280)$ to compute equal-risk prices numerically.% as expressed in Equation (\ref{equal_risk_price_S}).
\begin{figure}[htbp]
  \centering
  \includegraphics[width=0.45\textwidth]{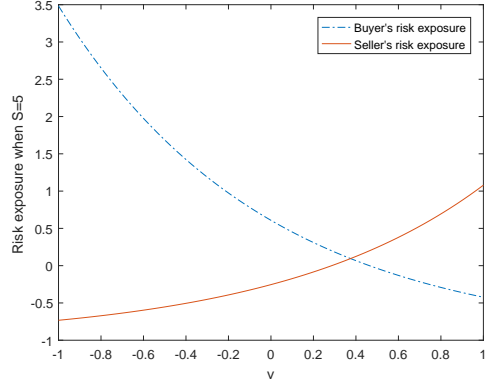}\\
  \caption{The minimum risk exposure for the  seller and the buyer with $S=5$.}\label{sample}
\end{figure}

In Figure \ref{sample}, we demonstrate how the minimum risk exposure for the buyer and seller changes as $v$ varies with $S=5$.  As expected, the minimum risk exposure for the seller is increasing; while the one for the buyer is decreasing as $v$ increases. The equal-risk price of a European call option with the current price $S=5$  corresponds to the offer price $v$ that makes $\rho^s(S,v;Z)=\rho^b(S,v;Z)$, which is numerically solved according to formula (\ref{equal_risk_price_S}).
\begin{figure}[!htbp]
\centering \subfigure[The equal-risk prices of European call options.]{
\label{figure4a}
\includegraphics[width=0.475\textwidth]{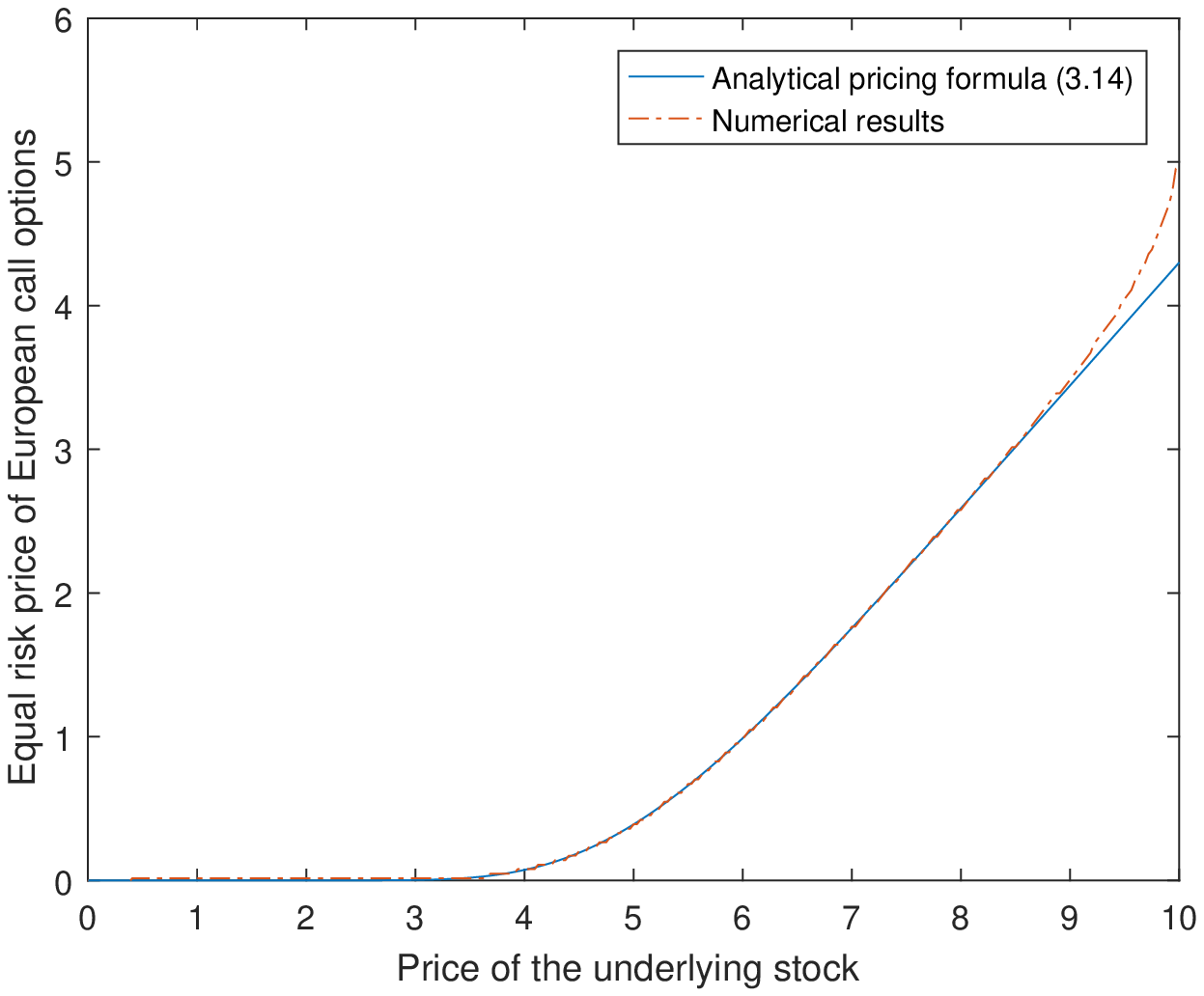}
} \subfigure[The absolute error. ] {
\label{figure4b}
\includegraphics[width=0.475\textwidth]{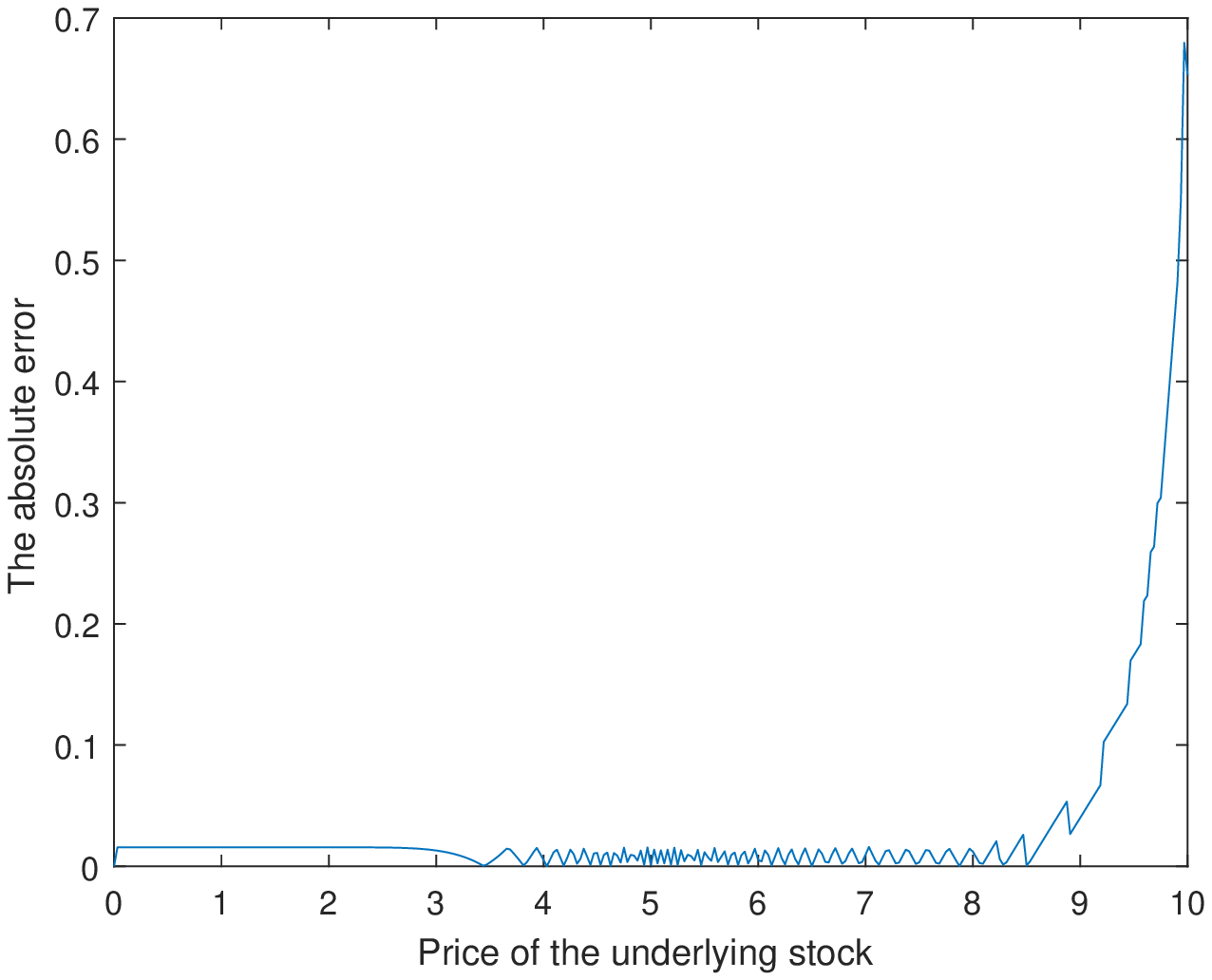}
}
\caption{Comparisons between the  pricing formula (\ref{equal_risk_price_call}) and our numerical results.}
\end{figure}

We repeat the above steps for different values of $S$.  The equal-risk prices of European call options  are plotted in Figure \ref{figure4a}  as the underlying stock price varies from $0$ to $10$, compared with the results  from the  pricing formula (\ref{equal_risk_price_call}). The absolute errors between them are plotted in Figure \ref{figure4b}. From Figures \ref{figure4a} and \ref{figure4b}, our numerical equal-risk prices are in good agreement with those from the  pricing formula except near the boundary $S=S_{\max}$. This error is the result of our approximate boundary condition at the truncated boundary. The first example demonstrates that our method for producing equal-risk prices by solving HJB equations numerically is consistent with the pricing formula, which provides motivates us to apply it to  general contingent claims in the next subsection.

\subsubsection{Example 2: Butterfly spread option}
The second example derives the equal-risk price for a butterfly spread option, of which the payoff is defined by
\begin{equation}\label{butterfly}
  Z(S)=(S-K_1)^+-2(S-\frac{K_1+K_2}{2})^++(S-K_2)^+.
\end{equation}
Figure \ref{figure_butterfly} provides a diagram of  the payoff.
\begin{figure}[htbp]
  \centering
  \includegraphics[width=0.45\textwidth]{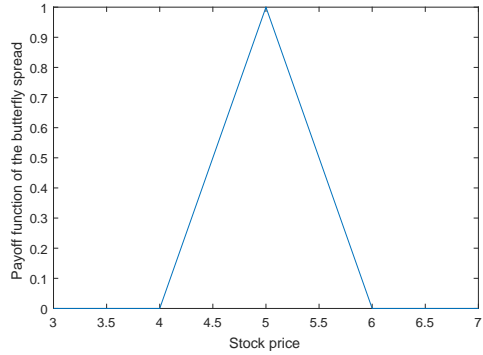}\\
  \caption{Payoff of a butterfly option with $K_1=4,K_2=6$.}\label{figure_butterfly}
\end{figure}
It is clear that the payoff is non-monotonic and non-smooth. \cite{guo2017equal} did not provide an equal-risk pricing formula for these cases. Since the payoff  is bounded, Assumption \ref{assumption1} always holds for any risk function. As a result, the optimal control problems (\ref{seller_risk}) and (\ref{buyer_risk})  are both well-defined and have a finite infimum. We now apply our numerical scheme to solve the these HJB equations first and then derive its equal-risk price numerically.

For the boundary conditions, since the stock follows a geometric Brownian motion, the butterfly spread option becomes worthless  at both $S=0$ and $S\rightarrow\i$. The seller  faces no liability and he has no motivation to hedge. Consequently, he would invests his initial wealth on the risk-free account and obtains the profits $ve^{r(T-t)}$ at time $T$. Consequently, the boundary conditions at  $S=0$ and $S\rightarrow\i$  are given by
\beq
\left\{
\begin{array}{l}
F^s(t,0,v)=R(-ve^{r(T-t)}), \\
\lim\limits_{S\rightarrow\i}F^s(t,S,v)=R(-ve^{r(T-t)}).\\
\end{array}\right.
\eeq
On the other hand, at the boundaries $S=0$ and $S\rightarrow\i$, the buyer pays $v$ at time $t$ for a worthless contingent claim and has no motivation to hedge. At expiry, the buyer only faces a deterministic liability $ve^{r(T-t)}$ and we impose the boundary conditions as
\beq
\left\{
\begin{array}{l}
F^b(t,0,v)=R(ve^{r(T-t)}), \\
\lim\limits_{S\rightarrow\i}F^b(t,S,v)=R(ve^{r(T-t)}).\\
\end{array}\right.
\eeq
The boundary condition along the $v$ direction are also implied by Lemma \ref{lemma1}, i.e
 \beq
\left\{
\begin{array}{l}
 \lim\limits_{v\rightarrow\i}F^s(t,S,v)=-1, \\
  \lim\limits_{v\rightarrow\i}F^b(t,S,v)=\i, \\
   \lim\limits_{v\rightarrow-\i}F^s(t,S,v)=\i, \\
\lim\limits_{v\rightarrow-\i}F^b(t,S,v)=-1,
\end{array}\right.
\eeq
which are approximated by
 \beq
\left\{
\begin{array}{l}
 F^s(t,S,v_{\max})=-1,\\
 F^b(t,S,v_{\max})=R(v_{\max}e^{r(T-t)}-Z(S)), \\
 F^s(t,S,-v_{\max})=R(Z(S)+v_{\max}e^{r(T-t)}), \\
 F^b(t,S,v_{\max})=-1,
\end{array}\right.
\eeq
to ensure their consistency with the terminal condition.

Now we are in the position to apply our numerical scheme to numerically solve the PDE system associated with the butterfly spread option.  The parameters in the second example are listed in Table \ref{parameters2}
\begin{table}[!htb]
\centering
\begin{tabular}{|c|c|c|c|c|c|c|c|c|c|}
  \hline
    Parameters& $K_1$  & $K_2$ & $T$ & $\mu$ &  $r$ &$\sigma$ & $S_{\max} $&$ v_{\max}$&$v_0$  \\
 \hline
  values&4 &6 &0.5 &0.05 & 0.05 & 0.3 &10 & $3$&1  \\
  \hline
\end{tabular}
\caption{Parameters.}
\label{parameters2}
\end{table}

For the butterfly spread option, we do not have an analytical solution. Hence we choose the results computed on the uniform mesh with $321\times 321\times 2560$ nodes as a benchmark solution. The numerical results of the value functions $F^s(T,S,v_0)$ and $F^b(T,S,v_0)$ calculated on different meshes are reported in Tables \ref{table6} and \ref{table7}.

\begin{table}[htbp]
\centering
\begin{tabular}{|l|c|l|l|l|l|l|c|c|c|}
\hline
$(N_x,N_y,N_T)$&$S=4$ &$S=4.5$&$S=5$ &$S=5.5$&$S=6$&$l_2$ error&ratio\\
\hline
(11,11,40) & -0.5654  & -0.4601 &  -0.3767 &  -0.4398 &  -0.4925&0.1183&\\
(21,21,80) & -0.5429  & -0.4816   &-0.4548  & -0.4715  & -0.5106&0.0294&4.0\\
 (41,41,160)& -0.5445&   -0.4919&   -0.4696&   -0.4832 &  -0.5172&0.0068&4.3\\
(81,81,320)& -0.5451  & -0.4944  & -0.4729&   -0.4859 &  -0.5189 &  0.0015  &4.7 \\
\hline
(321, 321, 2560)&  -0.5453  & -0.4951  &  -0.4739 &   -0.4867    &-0.5194&   &\\
\hline
\end{tabular}
\caption{The values of $F^s(T,S,v_0)$ on different meshes  }\label{table6}
\end{table}
\begin{table}[htbp]
\centering
\begin{tabular}{|l|c|l|l|l|l|l|c|c|c|}
\hline
$(N_x,N_y,N_T)$&$S=4$ &$S=4.5$&$S=5$ &$S=5.5$&$S=6$&$l_2$ error&ratio\\
\hline
(11,11,40) & -0.5480  &  -0.4491  &  -0.3658  &  -0.4112 &   -0.4385&0.1368&\\
(21,21,80) &  -0.5427&    -0.4807  &  -0.4508&    -0.4568 &   -0.4669&0.0324 &4.2\\
 (41,41,160)&  -0.5444  &  -0.4914   & -0.4665&    -0.4704  &  -0.4763&0.0071&4.5\\
(81,81,320)& -0.5451 &   -0.4939 &   -0.4701 &   -0.4736 &   -0.4786& 0.0013    &5.6 \\
%(161,161,640)& 1.3010   &   1.0932&      1.0145  &    1.0671  &    1.2015&  0.00034  &5&17.87 \\
\hline
(321, 321, 2560)& -0.5452   & -0.4946   & -0.4710 &   -0.4742 &   -0.4786&   &\\
\hline
\end{tabular}
\caption{The values of $F^b(T,S,v_0)$ on different meshes }\label{table7}
\end{table}
The $l_2$ error reported in Tables \ref{table6} and \ref{table7} indicates that the numerical results have converged and they can be used to produce the equal-risk price for the butterfly spread option by solving  (\ref{equal_risk}). Given $S=5$, we plot the minimum risk exposure for the seller and the buyer in Figure \ref{sample_b}. The equal-risk price for the butterfly spread option with the current price $S=5$ should be the price such that the minimum risk exposures for the seller and the buyer are equal. It can be numerically solved by formula (\ref{equal_risk_price_S}).
\begin{figure}[htbp]
  \centering
  \includegraphics[width=0.45\textwidth]{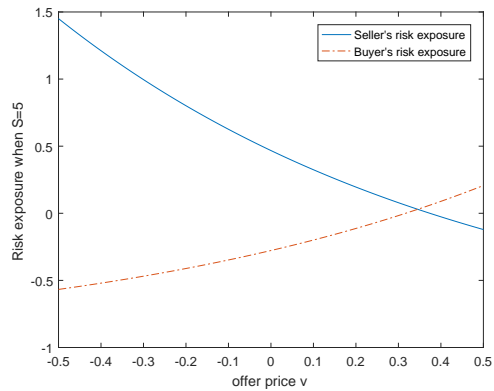}\\
  \caption{Equal-risk price of butterfly.}\label{sample_b}
\end{figure}

When  short selling is permitted and the market is complete,  a butterfly spread option can be replicated by three European call options as shown in Equation (\ref{butterfly}). Its Black-Scholes price is a linear combination of three call option prices.
\begin{equation}\label{butterfly_spread_BS}
  v=C^{BS}(S,K_1,r,\sigma,T)-2C^{BS}(S,\frac{K_1+K_2}{2},r,\sigma,T)+C^{BS}(S,K_2,r,\sigma,T).
\end{equation}
This Black-Scholes price is taken as the benchmark solution to illustrate the effect of the short selling ban on the butterfly spread option. The equal-risk prices calculated from our PDE framework and  those from the formula (\ref{butterfly_spread_BS}) are plotted in Figure \ref{figure7a} and the relative difference to the Black-Scholes price is depicted in Figure \ref{figure7b}.

\begin{figure}[!htbp]
\centering \subfigure[ Two kinds of price for the butterfly spread option.]{
\label{figure7a}
\includegraphics[width=0.475\textwidth]{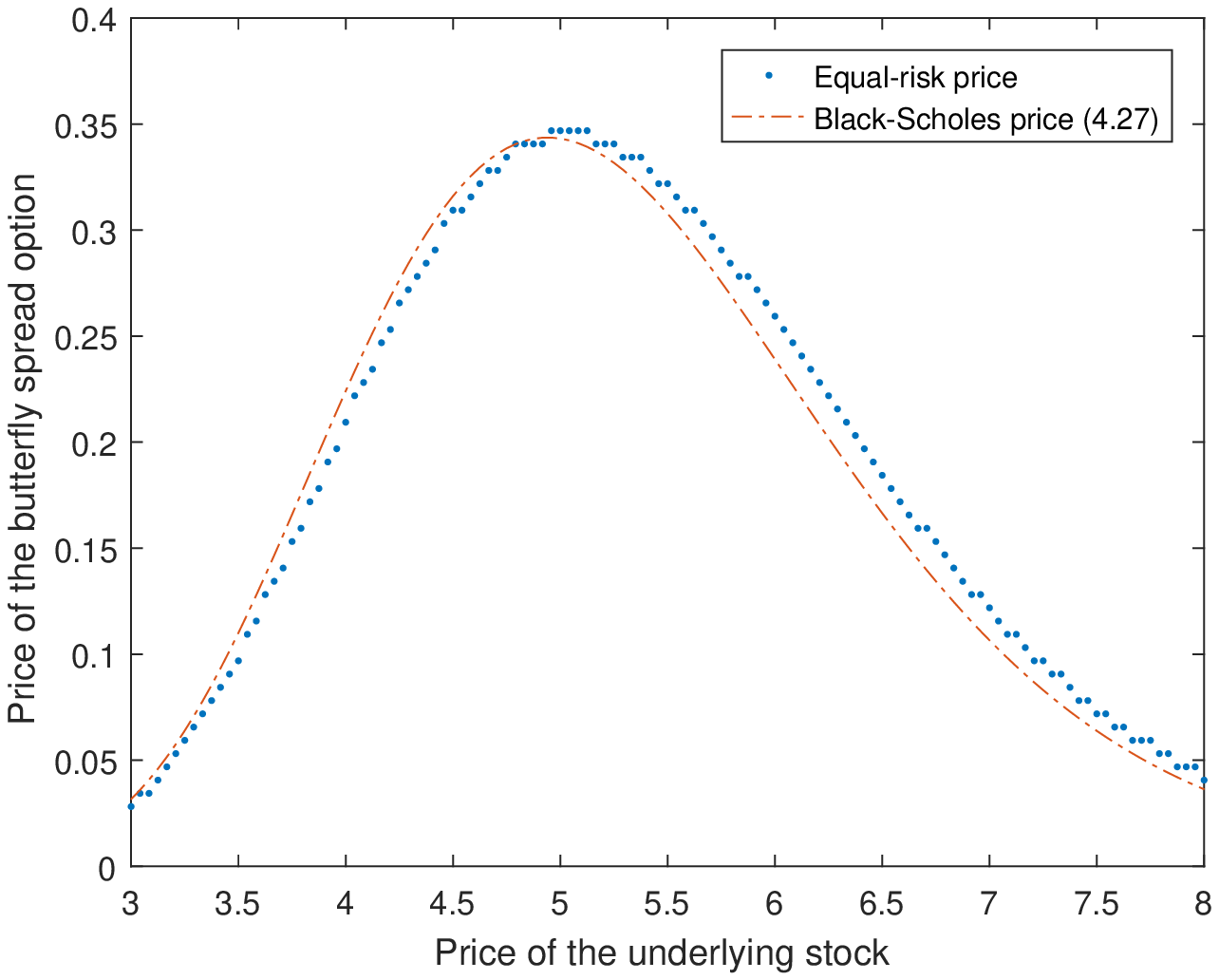}
} \subfigure[The percentage distance. ] {
\label{figure7b}
\includegraphics[width=0.475\textwidth]{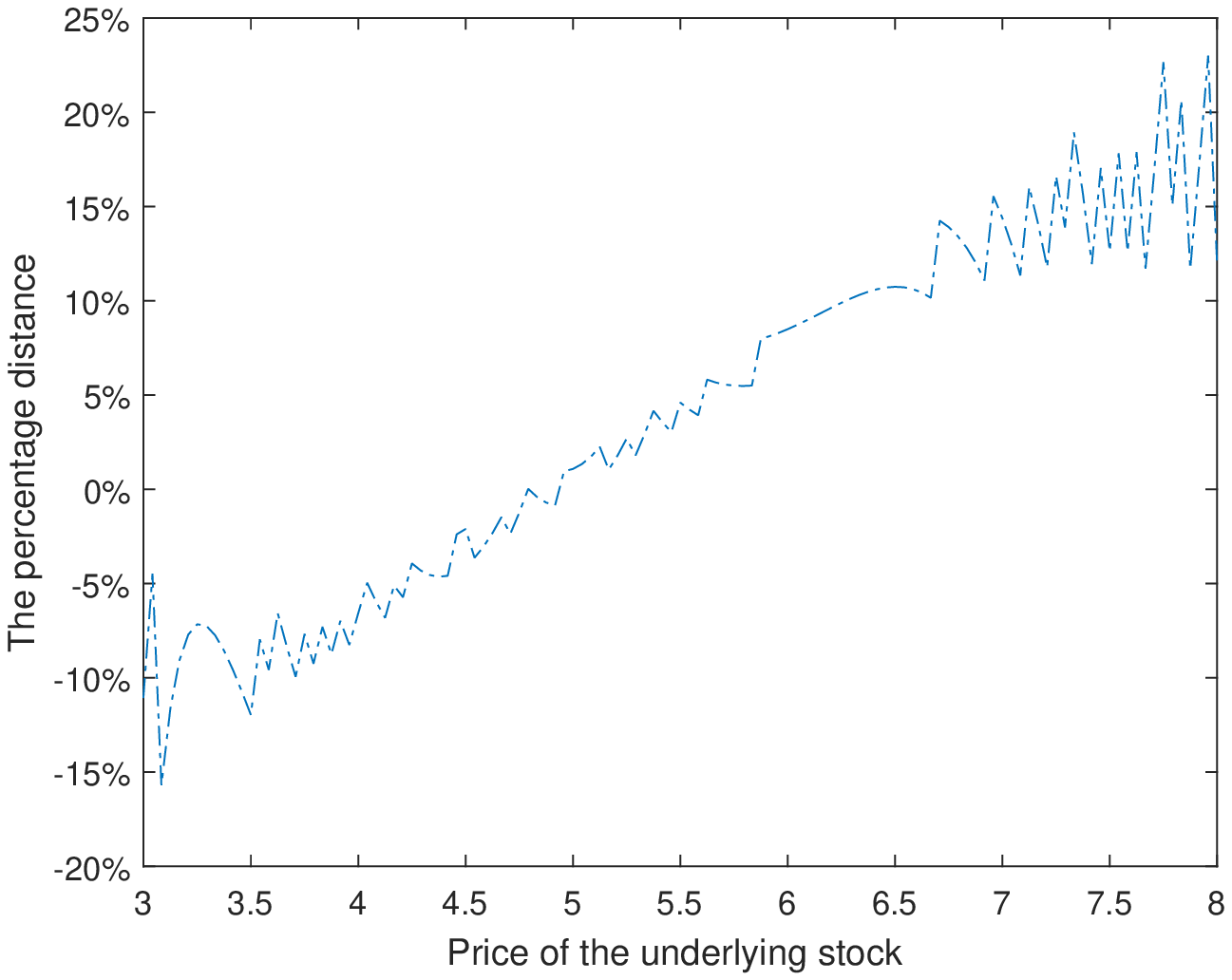}
}
\caption{Comparisons between equal-risk prices and Black-Scholes prices.}
\end{figure}

Unlike the cases of European call (put) options where short selling decreases (increases) the option price for all the underlying stock prices, it is observed from Figure \ref{figure7a} that the equal-risk price is higher than the Black-Scholes price when  $S>5$; while it is lower than the Black-Scholes price on the  other side.  Figure \ref{figure7b} shows that the relative difference between the equal-risk price and the Black-Scholes price is significant even though the absolute difference is small, which demonstrates that the short selling ban indeed affects the price of the butterfly spread option. In particular, the short selling ban lowers the option price in regions where the price is an increasing function of the underlying, and it raises the option prices in regions where the price is a decreasing function of the underlying.

Finally, we consider how the hedging strategy is affected by the  short selling ban, using the seller  as an example. The optimal hedging strategy for the seller is numerically calculated from the PDE system (\ref{ex1}). For comparison, the optimal hedging strategy in the Black-Scholes model without the short selling ban is
\begin{equation}\label{delta_butterfly_spread_BS}
  \phi^{BS}=\frac{\partial C^{BS}(S,K_1,r,\sigma,T)}{\partial S}-2\frac{\partial C^{BS}(S,\frac{K_1+K_2}{2},r,\sigma,T)}{\partial S}+\frac{\partial C^{BS}(S,K_2,r,\sigma,T)}{\partial S}.
\end{equation}
The numerical results calculated from the PDE system and the formula (\ref{delta_butterfly_spread_BS}) are plotted in Figure \ref{figure8a} with $v=0.5$.

\begin{figure}[!htbp]
\centering \subfigure[ Optimal hedging strategy for the seller .]{
\label{figure8a}
\includegraphics[width=0.475\textwidth]{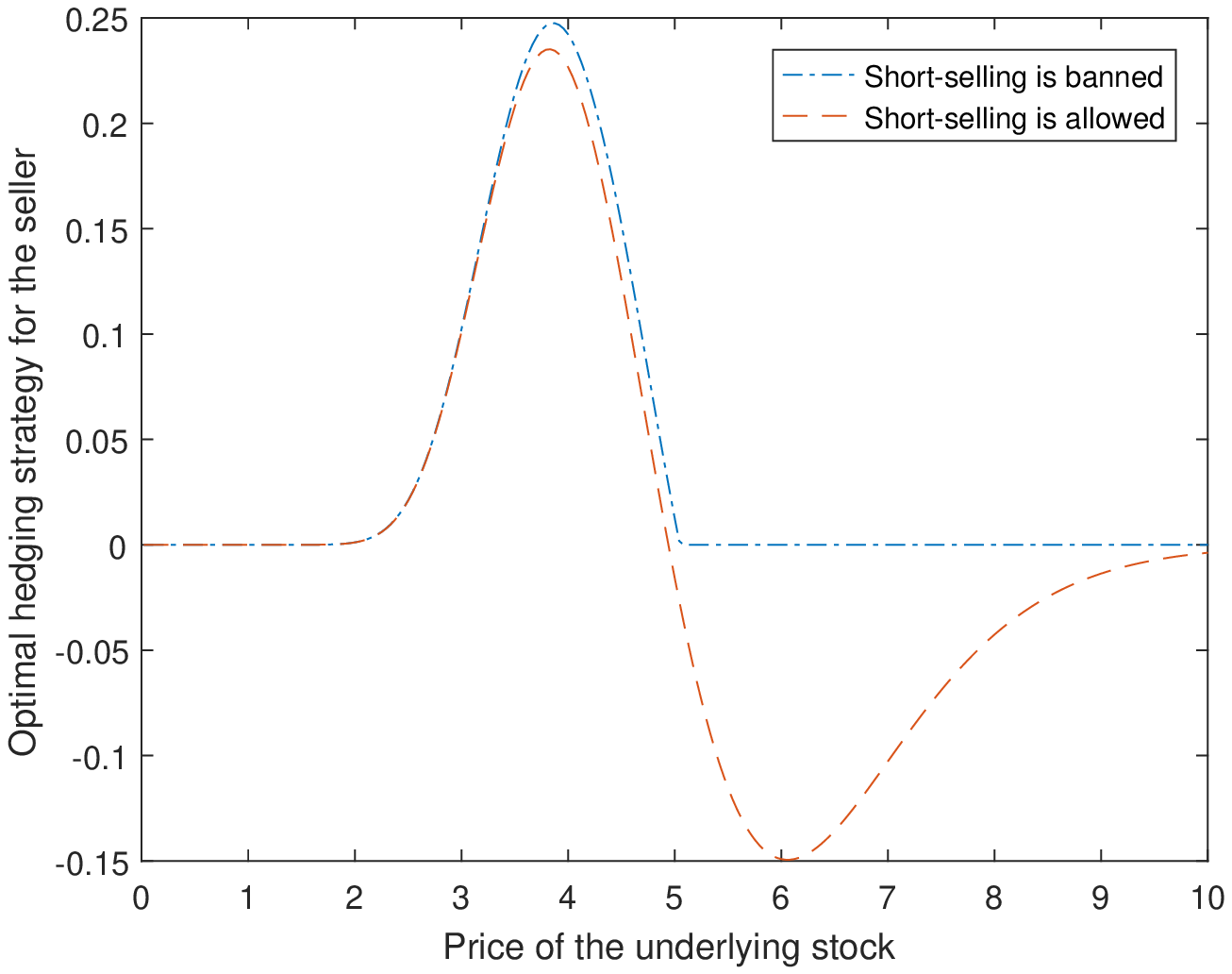}
} \subfigure[The absolute difference ] {
\label{figure8b}
\includegraphics[width=0.475\textwidth]{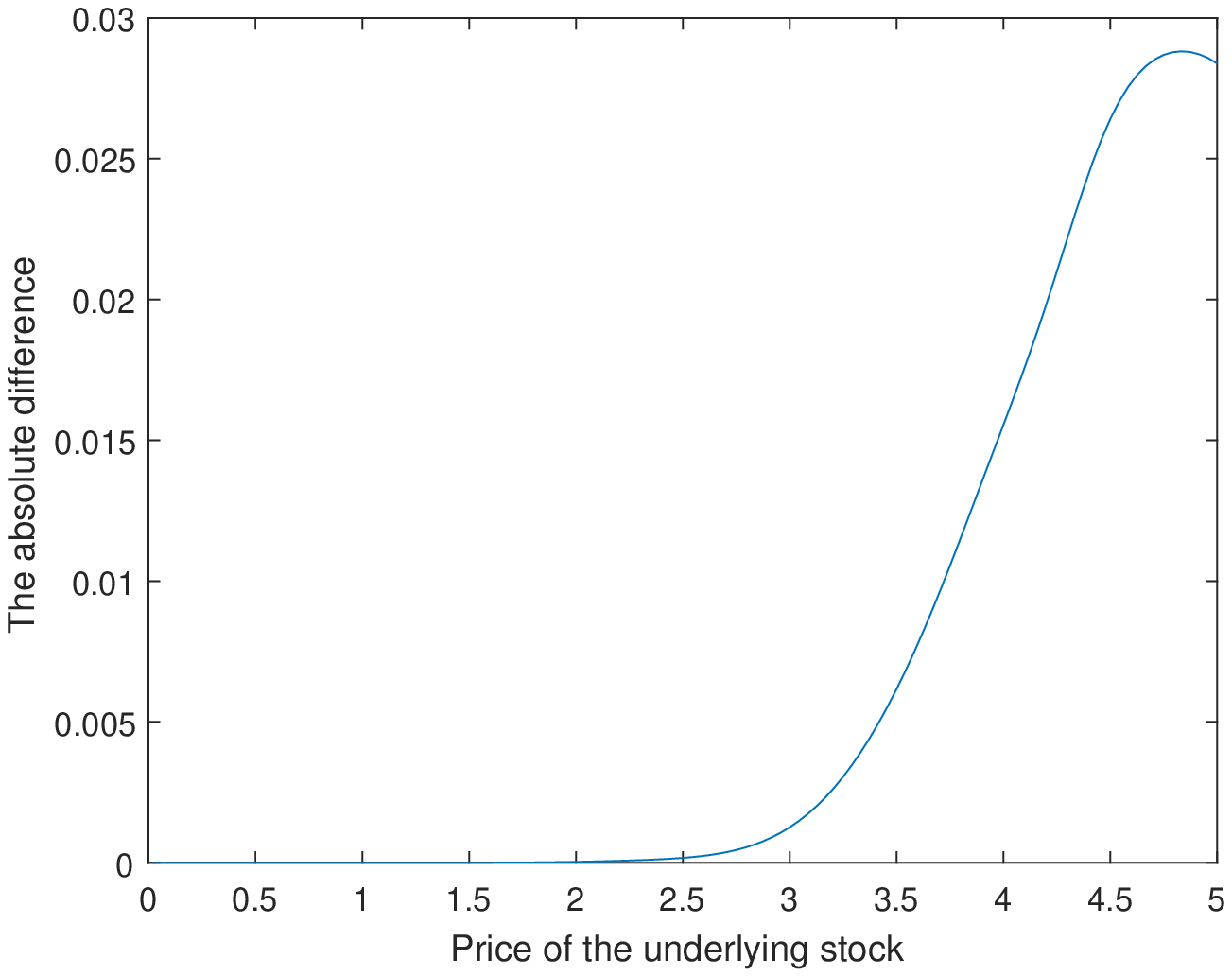}
}
\caption{Comparison between the optimal hedging strategy.}
\end{figure}

It is observed from Figure \ref{figure8a} that the optimal hedging strategy takes both positive and negative values as the underlying stock price varies when short selling is allowed. After imposing the short selling ban, the negative part becomes zero and the positive part increases. The absolute difference between them is plotted in Figure \ref{figure8b} when $S<5$.
\section{Conclusions}
This paper establishes a unified PDE framework for the recently proposed  \textit{equal-risk pricing approach} in order to explore how a short selling ban would affect the valuation of general contingent claims.   When the contingent claim is a European call or put option, the  PDE system can be  solved analytically, which leads to the same pricing formula provided by \cite{guo2017equal}. In addition, our PDE approach was able to adapt to options with  non-monotonic payoffs, such as the butterfly spread option, which was not addressed in \cite{guo2017equal}. Thus our PDE framework has significantly expanded the range of the application for the equal-risk pricing approach. According to the numerical results, the effects of the short selling ban are illustrated through comparisons between equal-risk prices and Black-Scholes prices. Generally, the short selling ban lowers the prices of  European call
options; while it has an opposite effect on the prices of  European put options. As for the butterfly spread option, the short selling ban lowers the option price when the payoff  is increasing with respect to the underlying stock price; while it raises the option price when the payoff  is decreasing.

\section{Acknowledgment}
Guiyuan Ma acknowledges the financial support from the National Natural Science Foundation of China (72101199) and the Fundamental Research Funds for the Central Universities (SK2021019). Song-Ping Zhu and Ivan Guo  acknowledge the financial support from the ARC under the ARC(DP) funding scheme (DP140102076, DP170101227). The Monash Centre for Quantitative Finance and Investment Strategies has been supported by BNP Paribas.
%The first author is also indebted to the University of Wollongong for providing a full PhD scholarship, which allowed him to be involved in this ARC-funded project as part of his PhD training process.

\begin{appendices}

\section{The proof of Lemma \ref{lemma1}}\label{appendix_1}
1. If $Z_1\leq Z_2$, the following inequality always holds because of the monotonicity of $R(x)$ for any admissible hedging strategy $\phi(\cdot)$,
\begin{equation}
   R\left(Z_1(S_T)-v_T^{v,\phi(\cdot)}\right)\leq R\left(Z_2(S_T)-v_T^{v,\phi(\cdot)}\right).
\end{equation}
Taking expectation and infimum on both sides leads to
\begin{equation}
  \rho^s(S,v;Z_1)\leq \rho^s(S,v;Z_2).
\end{equation}
When $v_1\leq v_2$, for any admissible hedging strategy $\phi(\cdot)$ the inequality becomes
\begin{equation}
   R\left(Z(S_T)-v_T^{v_1,\phi(\cdot)}\right)\geq R\left(Z(S_T)-v_T^{v_2,\phi(\cdot)}\right),
\end{equation}
which results in
\begin{equation}
  \rho^s(S,v_1;Z)\geq \rho^s(S,v_2;Z).
\end{equation}
By relation (\ref{equivalence}), the monotonicity of $\rho^b(S,v;Z)$ is characterized as
\begin{align}
   \rho^b(S,v;Z_1)&=\rho^s(S,-v;-Z_1)\geq \rho^s(S,-v;-Z_2)= \rho^b(S,v;Z_2)\nn\\
\rho^b(S,v_1;Z)&=\rho^s(S,-v_1;-Z)\leq \rho^s(S,-v_2;-Z)= \rho^b(S,v_2;Z)\nn.
\end{align}
2. Choosing any admissible strategy $\phi$ satisfying Assumption \ref{assumption1}, we obtain
 \begin{equation}\label{up}
 \lim\limits_{v\rightarrow\i}\rho^s(S,v;Z)\leq
  \lim\limits_{v\rightarrow\i}\mathbf{E}_{\mathbb{Q}}\left[R\left(Z(S_T)-v_T^{v,\phi}\right)\right]\triangleq L^B.
 %\lim\limits_{v\rightarrow\i}\mathbf{E}_{\mathbb{Q}}R(Z(S_T)-v_T^{v,\phi})=\lim\limits_{v\rightarrow\i}\mathbf{E}_{\mathbb{Q}}R(Z(S_T)-ve^{rT})=\textit{LB}.
 \end{equation}
Due to the fact that $R\left(Z(S_T)-v_T^{v,\phi(\cdot)}\right)\geq L^B$ always holds for any  $\phi(\cdot)\in\Phi$, we have
\begin{equation}\label{up2}
   \lim\limits_{v\rightarrow\i}\rho^s(S,v;Z)\geq L^B.
\end{equation}
Combing Equations (\ref{up}) and  (\ref{up2}) together, we have $\lim\limits_{v\rightarrow\i}\rho^s(S,v;Z)=L^B$.

For any $\phi(\cdot)\in\Phi$, we apply Jensen's inequality to risk function $R(x)$ and obtain
\begin{equation}\label{up1}
 \mathbf{E}_{\mathbb{Q}}^{v,S}\left[R\left(v_T^{v,-\phi(\cdot)}-Z(S_T)\right)\right]\geq R\left(ve^{rT}-\mathbf{E}_{\mathbb{Q}}Z(S_T)\right).
\end{equation}
Taking infimum and limits on both sides results in
\begin{equation*}
\lim\limits_{v\rightarrow\i}\rho^b(S,v;Z) =\lim\limits_{v\rightarrow\i}\inf\limits_{\phi(\cdot)\in \mathbf{\Phi}}  \mathbf{E}_{\mathbb{Q}}^{v,S}\left[R\left(v_T^{v,-\phi(\cdot)}-Z(S_T)\right)\right]\geq \lim\limits_{v\rightarrow\i}R\left(ve^{rT}-\mathbf{E}_{\mathbb{Q}}Z(S_T)\right)=\i.
 \end{equation*}
Following the relation (\ref{equivalence}), it is easy to derive that
\begin{align}
    \lim\limits_{v\rightarrow -\i}\rho^s(S,v;Z)&=\lim\limits_{v\rightarrow -\i}\rho^b(S,-v;-Z)=\lim\limits_{v\rightarrow \i}\rho^b(S,v;-Z)=\i\nn\\
\lim\limits_{v\rightarrow -\i}\rho^b(S,v;Z)&= \lim\limits_{v\rightarrow -\i}\rho^s(S,-v;-Z)= \lim\limits_{v\rightarrow \i}\rho^s(S,v;-Z)=L^B\nn,
\end{align}
which completes the proof.
\section{The proof of Theorem \ref{thm1}}\label{appendix_2}
Given the current underlying price $S$ and the European contingent claim $Z$, we construct a map:
\begin{equation}\label{map}
  H(v):=\rho^b(S,v;Z)-\rho^s(S,v;Z).
\end{equation}
According to Lemma \ref{lemma1}, such a map $H(v)$ is continuous and non-decreasing. On one hand, we have
\begin{equation}\label{v0}
  \lim\limits_{v\rightarrow -\i}H(v)=\lim\limits_{v\rightarrow-\i}[\rho^b(S,v;Z)-\rho^s(S,v;Z)]=-\i.
\end{equation}
On the other hand, as $v$ tends toward infinity, we obtain
\begin{equation}\label{vi}
  \lim\limits_{v\rightarrow\i}H(v)=\lim\limits_{v\rightarrow\i}[\rho^b(S,v;Z)-\rho^s(S,v;Z)]=\i.
\end{equation}
Hence, we conclude that there exists at least one solution to $H(v)=0$ on $(-\i,\i)$.

To demonstrate the uniqueness of the solution, we first assume that the equation $H(v)=0$ has two different solutions $v_1>v_2$. According to the monotonicity described in Lemma \ref{lemma1}, we have
\begin{equation}\label{ineq}
  \rho^b(S,v_1;Z)\geq\rho^b(S,v_2;Z)=\rho^s(S,v_2;Z)\geq\rho^s(S,v_1;Z)=\rho^b(S,v_1;Z),
\end{equation}
which implies that $\rho^b(S,v_1;Z)=\rho^b(S,v_2;Z)$. Again, according to the monotonicity and convexity of $\rho^b(S,v;Z)$ with respect to $v$, we come to a conclusion that $\rho^b(S,v;Z)$ is constant for $ v\leq v_1$. It follows that
\begin{equation}\label{fow}
  \rho^s(S,v_1;Z)=\rho^b(S,v_2;Z)=\lim\limits_{v\rightarrow-\i}\rho^b(S,v;Z)=L^B
\end{equation}
By Jensen's inequality, we have
\beq
\left\{
\begin{array}{l}
\ds  R(v_1e^{rT}-\mathbf{E}_{\mathbb{Q}}Z(S_T))\leq\rho^s(S,v_1;Z)=L^B\leq 0,  \label{jensen} \\
  \ds R(\mathbf{E}_{\mathbb{Q}}Z(S_T)-v_2e^{rT})\leq\rho^b(S,v_2;Z)=L^B\leq 0.
\end{array}\right.
\eeq
The above equations implies that both $v_1e^{rT}-\mathbf{E}_{\mathbb{Q}}\left[Z(S_T)\right]$ and $\mathbf{E}_{\mathbb{Q}}\left[Z(S_T)\right]-v_2e^{rT}$ are non-positive because that $R(x)\geq 0$ for any $x\geq 0$. However, this conclusion contradicts the fact that
\begin{equation}\label{contradict}
  v_1e^{rT}-\mathbf{E}_{\mathbb{Q}}\left[Z(S_T)\right]+\mathbf{E}_{\mathbb{Q}}\left[Z(S_T)\right]-v_2e^{rT}=(v_1-v_2)e^{rT}>0.
\end{equation}
Therefore, the solution must be unique.
\section{The proof of Corollaries  \ref{coro1}-\ref{coro3}}\label{appendix_3}
For Corollary \ref{coro1},
we consider the  minimum risk exposure of the seller for a contingent claim $-(K-S)^+$ first. To calculate $\rho^s(S,v;-(K-S)^+)$, we need to  solve the associate HJB equation
\beq
\left\{
\begin{array}{l}
\ds  0=\frac{\partial F^s}{\partial t}+\inf\limits_{\phi\geq 0}\left\{\mathcal{L}_1^{\phi}F^s\right\},\label{HJB_seller1} \\
  \ds F^s(T,S,v)=R\left(-(K-S)^+-v\right).
\end{array}\right.
\eeq
With the same technique in Proposition \ref{proposition1}, the solution can be derived as
 \begin{equation}\label{put}
  F^s(t,S,v)=R\left(e^{r(T-t)}[-P^{BS}(S,K,r,\sigma,T-t)-v]\right).
 \end{equation}
According to the relation (\ref{equivalence}), we have
\begin{equation*}
  \rho^b(S,v;(K-S)^+)=\rho^s\left(S,-v;-(K-S)^+\right)=F^s(0,S,-v)=R\left(e^{rT}[v-P^{BS}(S,K,r,\sigma,T)]\right).
\end{equation*}
For Corollary \ref{coro2}, we explore
the minimum  risk exposure of the buyer for a contingent claim $-(K-S)^+$ first. To compute $\rho^b(S,v;-(K-S)^+)$, we goes to the HJB equation
\beq
\left\{
\begin{array}{l}
\ds  0=\frac{\partial F^b}{\partial t}+\inf\limits_{\phi\geq 0}\left\{\mathcal{L}_2^{\phi}F^b\right\},\label{HJB_seller2} \\
  \ds F^b(T,S,v)=R\left(v+(K-S)^+\right).
\end{array}\right.
\eeq
Similar to Proposition \ref{proposition2}, the solution to such a PDE system is
 \begin{equation}\label{coroll}
   F^b(t,S,v)=\frac{1}{\sqrt{2\pi}}\int_{-\i}^{\i}R\left(ve^{r(T-t)}+(K-Se^{(r-\frac{\sigma^2}{2})(T-t)+\sigma\sqrt{T-t}x})^+\right)e^{-\frac{x^2}{2}}dx.
 \end{equation}
From relation (\ref{equivalence}), the seller's risk exposure of European put options is
\begin{align*}
   \rho^s(S,v;(K-S)^+)&=\rho^b(S,-v;-(K-S)^+)=F^b(0,S,-v)\nn\\
&=\frac{1}{\sqrt{2\pi}}\int_{-\i}^{\i}R\left((K-Se^{(r-\frac{\sigma^2}{2})T+\sigma\sqrt{T}x})^+-ve^{rT}\right)e^{-\frac{x^2}{2}}dx.\label{demo4}
\end{align*}
Finally, the proof of Corollary \ref{coro3} is similar to Theorem \ref{thm1}.
\end{appendices}

\pagestyle{plain}
\setlength{\bibsep}{0ex}
\begin{spacing}{1}
\bibliographystyle{apa}
\bibliography{ref}
\end{spacing}

\end{document}